\documentclass[envcountsame,runningheads,a4paper]{llncs}
\usepackage[colorlinks=true, citecolor=blue,linkcolor=red,urlcolor=black]{hyperref}
\usepackage[utf8]{inputenc}
\usepackage{hhline}
\usepackage{fullpage}

 \usepackage{multirow}
 
\usepackage{times}

\usepackage{xcolor,tikz,pgf,pgffor}
\usetikzlibrary{automata,shapes,backgrounds,decorations.pathreplacing}

\usepackage{amsmath}
\usepackage{amssymb,stmaryrd}
\usepackage{enumitem}
\usepackage{multirow}
\usepackage{wrapfig}
\usetikzlibrary{decorations.pathreplacing}
\usetikzlibrary{matrix}

\usepackage[mathlines]{lineno}
\let\oldqed\qed
\renewcommand\qed{\mbox{}\hfill$\oldqed$}
\newcommand*\patchAmsMathEnvironmentForLineno[1]{%
  \expandafter\let\csname old#1\expandafter\endcsname\csname #1\endcsname
  \expandafter\let\csname oldend#1\expandafter\endcsname\csname end#1\endcsname
  \renewenvironment{#1}%
  {\linenomath\csname old#1\endcsname}%
  {\csname oldend#1\endcsname\endlinenomath}}%
\newcommand*\patchBothAmsMathEnvironmentsForLineno[1]{%
  \patchAmsMathEnvironmentForLineno{#1}%
  \patchAmsMathEnvironmentForLineno{#1*}}%
\AtBeginDocument{%
  \patchBothAmsMathEnvironmentsForLineno{equation}%
  \patchBothAmsMathEnvironmentsForLineno{align}%
  \patchBothAmsMathEnvironmentsForLineno{flalign}%
  \patchBothAmsMathEnvironmentsForLineno{alignat}%
  \patchBothAmsMathEnvironmentsForLineno{gather}%
  \patchBothAmsMathEnvironmentsForLineno{multline}%
}

\makeatletter
\let\c@definition\c@theorem
\let\c@lemma\c@theorem
\let\c@corollary\c@theorem
\let\c@remark\c@theorem
\let\c@example\c@theorem
\let\c@proposition\c@theorem
\let\c@problem\c@theorem
\makeatother

\SetSymbolFont{stmry}{bold}{U}{stmry}{m}{n}

\usepackage{algorithm}
\usepackage[noend]{algpseudocode}

\algnewcommand\algorithmicforeach{\textbf{for each}}
\algdef{S}[FOR]{ForEach}[1]{\algorithmicforeach\ #1\ \algorithmicdo}

\newcommand{\N}{\mathbb{N}}

\newcommand{\ssetminus}{\! \setminus \!}

\def\rest#1#2{#1_{\restriction#2}}
\def\subgame#1#2{#1 \setminus #2}

\newcommand{\Plays}{\mathsf{Plays}}

\newcommand{\Obj}{\Omega}  
\newcommand{\Win}[3]{{\mathsf{Win}}({#2},{#1},{#3})}  

\newcommand{\Occ}{\mathsf{Occ}}
\newcommand{\Occinf}{\mathsf{Inf}}

\newcommand{\Reach}{\mathsf{Reach}}
\newcommand{\Safe}{\mathsf{Safe}}
\newcommand{\Buchi}{\mathsf{B\ddot{u}chi}}
\newcommand{\CoBuchi}{\mathsf{CoB\ddot{u}chi}}

\newcommand{\GenBuchi}{\mathsf{GenB\ddot{u}chi}}
\newcommand{\Par}{\mathsf{EvenParity}}
\newcommand{\OddPar}{\mathsf{OddParity}}
\newcommand{\iPar}[1]{{#1}\mathsf{Parity}}
\newcommand{\ConjPar}{\mathsf{ConjEvenParity}}
\newcommand{\DisjPar}{\mathsf{DisjOddParity}}

\newcommand{\posAttr}{positive attractor}
\newcommand{\layeredAttr}{layered attractor}
\newcommand{\posSafeAttr}{positive safe attractor}
\newcommand{\CPre}[3]{{\mathsf{Cpre}}_{#1}({#2},{#3})}
\newcommand{\Attr}[3]{{\mathsf{Attr}}_{#1}({#2},{#3})}
\newcommand{\PosAttr}[3]{{\mathsf{PAttr}}_{#1}({#2},{#3})}
\newcommand{\PosSafeAttr}[4]{{\mathsf{PSafeAttr}}_{#1}({#2},{#3},{#4})}

\newcommand{\LAttr}[5]{{\mathsf{LayAttr}}_{#1}({#2},{#3},{#4},{#5})}

\newcommand{\JF}[3]{{\mathsf{GoodEp}}_{#1}({#2},{#3})}
\newcommand{\QJF}[4]{{\mathsf{LayEp}}_{#1}({#2},{#3},{#4})}

\newcommand{\Pset}{P_{\geq q}}
\newcommand{\p}{\alpha} 
\newcommand{\pmin}{q}
\newcommand{\pmax}[1]{d({#1})}

\newcommand{\pmini}[1]{q_{{#1}}}
\newcommand{\pmaxi}[1]{d_{{#1}}(0)}

\newcommand{\Zset}{Z}

\newcommand{\List}{L}
\newcommand{\GenList}{{\cal P}}
\newcommand{\psolB}{BüchiSolver}
\newcommand{\GenpsolB}{GenBüchiSolver}
\newcommand{\psolC}{GoodEpSolver}
\newcommand{\GenpsolC}{GenGoodEpSolver}
\newcommand{\psolQ}{LaySolver}
\newcommand{\GenpsolQ}{GenLaySolver}
\newcommand{\Ziel}{Zielonka}
\newcommand{\GenZiel}{GenZielonka}
\newcommand{\PartialZiel}{Ziel\&PSolver}
\newcommand{\GenPartialZiel}{GenZiel\&PSolver}

\newcommand{\up}{up}
\newcommand{\down}{down}
\newcommand{\init}{\p}
\newcommand{\edge}[1]{E_{{#1}}}

\title{Partial Solvers for Generalized Parity Games\thanks{Work partially supported by the PDR project \emph{Subgame perfection in graph games} (F.R.S.-FNRS), the ARC project \emph{Non-Zero Sum Game Graphs: Applications to Reactive Synthesis and Beyond} (Fédération Wallonie-Bruxelles), the EOS project \emph{Verifying Learning Artificial Intelligence Systems} (F.R.S.-FNRS \& FWO), and the COST Action 16228 \emph{GAMENET} (European Cooperation in Science and Technology)}}

\author{V\'eronique Bruy\`ere\inst{1} \and Guillermo A. P\'erez\inst{2} \and Jean-Fran\c cois Raskin\inst{3} \and Cl\'ement Tamines\inst{1}}
\authorrunning{V. Bruy\`ere, G. A. P\'erez, J.-F. Raskin, C. Tamines}
\institute{University of Mons (UMONS), Mons, Belgium \\
\email{\{veronique.bruyere,clement.tamines\}@umons.ac.be} \and
University of Antwerp (UAntwerp), Antwerp, Belgium\\
\email{guillermoalberto.perez@uantwerpen.be} \and
Universit\'e libre de Bruxelles (ULB), Brussels, Belgium \\
\email{jraskin@ulb.ac.be}
}

\begin{document}

\maketitle

\begin{abstract}
Parity games have been broadly studied in recent years for their applications to controller synthesis and verification. In practice, partial solvers for parity games that execute in polynomial time, while incomplete, can solve most games in publicly available benchmark suites. In this paper, we combine those partial solvers with the classical recursive algorithm for parity games due to Zielonka.
We also extend partial solvers to generalized parity games that are games with conjunction of parity objectives.  
We have implemented those algorithms and evaluated them on a large set of benchmarks proposed in the last LTL synthesis competition. 
\end{abstract}


\section{Introduction}

Since the early nineties, parity games have been attracting a large attention in the formal methods and theoretical computer science communities for two main reasons. First, parity games are used as intermediary steps in the solution of several relevant problems like, among others, the reactive synthesis problem from LTL specifications~\cite{PR89} or the emptiness problem for tree automata~\cite{EmersonJ91}. Second, their exact complexity is a long standing open problem: while we know that they are in ${\sf NP} \cap {\sf coNP}$~\cite{EmersonJ91} (and even in ${\sf UP} \cap {\sf coUP}$~\cite{Jurdzinski98}), we do not yet have a polynomial time algorithm to solve them. Indeed, the best known algorithm so far has a worst-case complexity which is quasi-polynomial~\cite{Calude}.

The classical algorithm for reactive synthesis from LTL specifications is as follows: from an LTL formula $\phi$ whose propositional variables are partitioned into inputs (controllable by the environment) and outputs (controlled by the system), construct a deterministic parity automaton (DPA) $A_\phi$ that recognizes the set of traces that are models of $\phi$. This DPA can then be seen as a two player graph game where the two players choose in turn the values of the input variables (Player 1) and of the output variables (Player 0). The winning condition in this game is the parity acceptance condition of the DPA. The two main difficulties with this approach are that the DPA may be doubly exponential in the size of  $\phi$ and that its parity condition may require exponentially many priorities. So the underlying parity game may be hard to solve with the existing algorithms (which are not polynomial time). These difficulties have triggered two series of results. 

First, incomplete algorithms that partially solve parity games in polynomial time have been investigated in~\cite{HuthKP13,HuthKP16,Ah-FatH16}. Although they are incomplete, experimental results show that they behave well on benchmarks generated with a random model and on examples that are forcing the worst-case behavior of the classical recursive algorithm for solving parity games due to Zielonka~\cite{zielonka98}. The latter algorithm has a worst-case complexity which is exponential in the number of priorities of the parity condition.
Second, compositional approaches to generate the automata from LTL specifications have been advocated, when the LTL formula $\phi$ is a conjunction of smaller formulas, i.e., $\phi=\phi_1 \land \phi_2 \cdots \land \phi_n$. In this case, the procedure constructs a DPA $A_i$ for each subformula $\phi_i$. The underlying game is then the product of the automata $A_i$ and the winning condition is the conjunction (for Player 0) of the parity conditions of each automaton. Those games are thus generalized parity games, that are known to be $\mathsf{co}$-$\mathsf{NP}$-complete~\cite{ChatterjeeHP07}.

In this paper, we contribute to these two lines of research in several ways. First, we show how to extend the partial solvers for parity games to generalized parity games. In the generalized case, we also show how efficient data structures based on antichains can be used to retain efficiency. Second, we show how to combine partial solvers for parity games and generalized parity games with the classical recursive algorithms~\cite{zielonka98,ChatterjeeHP07}. In this combination, the recursive algorithm is only executed on the portion of the game graph that was not solved by the partial solver, and this is repeated at each recursive call. 
Third, we provide for the first time extensive experiments that compare all those algorithms on benchmarks that are generated from LTL specifications used in the LTL synthesis competition~\cite{syntcomp18}. For parity games, our experiments show behaviors that differ largely from the behaviors observed on experiments done on random graphs only. Indeed Zielonka's algorithm is faster than partial solvers on average which was not observed on random graphs in~\cite{HuthKP16}. Equally interestingly, we show that there are instances of our benchmarks of generalized parity games that cannot be solved by the classical recursive algorithm or by any of the partial solvers alone, but that can be solved by algorithms that combine them. We also show that when combined with partial solvers, the performances of the classical recursive algorithms are improved on a large portion of our benchmarks for both the parity and generalized parity cases.

The structure of the paper is as follows. In Section~\ref{sec:preliminaries}, we recall the useful notions on two-player games played on graphs. In particular we recall the notions of parity game and generalized parity game.
In Section~\ref{sec:Ziel}, we explain how the classical recursive algorithms for parity games and generalized parity can be combined with partial solvers and under which hypothesis the resulting algorithm is correct. In Sections~\ref{sec:psolB}-\ref{sec:psolQ}, we present the three partial solvers proposed in \cite{HuthKP13,HuthKP16} for parity games and we explain how to extend them to generalized parity games. For the extended partial solver of Section~\ref{sec:psolC}, we also explain how to transform it into an antichain-based algorithm. In the last Section~\ref{sec:experiments}, we present our experiments that compare the three partial solvers for both parity games and generalized parity games.


\section{Preliminaries} \label{sec:preliminaries}

\paragraph{\bf Game structures.~} 

A \emph{game structure} is a tuple $G = (V_0,V_1,E)$ where
\begin{itemize}
\item $(V,E)$ is a finite directed graph, with $V = V_0 \cup V_1$ the set of vertices and $E \subseteq V \times V$ the set of edges such that\footnote{This condition guarantees that there is no deadlock.} for each $v \in V$, there exists $(v,v') \in E$ for some $v' \in V$,
\item $(V_0,V_1)$ forms a partition of $V$ such that $V_i$ is the set of vertices controlled by player~$i$ with $i \in \{0,1\}$.
\end{itemize}

Given $U \subseteq V$, if $\rest{G}{U} = (V_0 \cap U,V_1 \cap U,E \cap (U \times U))$ has no deadlock, then $\rest{G}{U}$ is called the  \emph{subgame structure} induced by $U$.

A \emph{play} in $G$ is an infinite sequence of vertices $\pi = v_0 v_1 \ldots \in V^{\omega}$ such that $(v_j,v_{j+1}) \in E$ for all $j \in \N$. \emph{Histories} in $G$ are finite sequences $h = v_0 \ldots v_j \in V^+$ defined in the same way. We denote by $\Plays(G)$ the set of plays in $G$ and by $\Plays(v_0)$ the set of plays starting in a given \emph{initial vertex} $v_0$.  
Given a play $\pi = v_0 v_1 \ldots$, the set $\Occ(\pi)$ denotes the set of vertices that occur in $\pi$, and the set $\Occinf(\pi)$ denotes the set of vertices that occur infinitely often in $\pi$, i.e., $\Occ(\pi) = \{v \in V \mid \exists j \geq 0, v_j = v \}$ and $\Occinf(\pi) = \{v \in V \mid \forall j \geq 0, \exists k \geq j,\ v_k = v\}$. 

\paragraph{\bf Strategies.~} 

A \emph{strategy} $\sigma_i$ for player~$i \in \{0,1\}$ is a function $\sigma_i\colon V^*V_i \rightarrow V$ assigning to each history $hv \in V^*V_i$ a vertex $v' = \sigma_i(hv)$ such that $(v,v') \in E$. It is \emph{memoryless} if $\sigma_i(hv) = \sigma_i(h'v)$ for all histories $hv, h'v$ ending with the same vertex $v$, that is, $\sigma_i$ is a function $\sigma_i\colon V_i \rightarrow V$. More generally, it is \emph{finite-memory} if $\sigma_i(hv)$ needs only a finite information out of the history $hv$. This is possible with a finite-state machine that keeps track of histories of plays (see~\cite{2001automata} for a precise definition).
%
Given a strategy $\sigma_i$ of player~$i$, a play $\pi = v_0 v_1 \ldots$ of $G$ is \emph{consistent} with $\sigma_i$ if $v_{j+1} = \sigma_i(v_0 \ldots v_j)$ for all $j \in \N$ such that $v_j \in V_i$. Consistency is naturally extended to histories in a similar way. 

\paragraph{\bf Objectives.~}

An \emph{objective for player~$i$} is a set of plays $\Obj \subseteq \Plays(G)$. It is called \emph{prefix-independent} whenever $\pi \in \Obj$ if and only if $h\pi \in \Obj$ for all plays $\pi, h\pi$. A \emph{game} $(G,\Obj)$ is composed of a game structure $G$ and an objective~$\Obj$ for \emph{player~$0$}. A play $\pi$ is \emph{winning} for player~$0$ if $\pi \in \Obj$, and losing otherwise. The games that we here study are \emph{zero-sum}: player~$1$ has the opposite objective $\overline{\Obj} = V^{\omega} \setminus \Obj$, meaning that a play $\pi$ is winning for player~$0$ if and only if it is losing for player~$1$. Given a game $(G,\Obj)$ and an initial vertex $v_0$, a strategy $\sigma_0$ for player~$0$ is \emph{winning from} $v_0$ if all the plays $\pi \in \Plays(v_0)$ consistent with $\sigma_0$ belong to $\Obj$. 
Vertex $v_0$ is thus called \emph{winning} for player~$0$. We also say that player~$0$ is winning from $v_0$. We denote by $\Win{0}{G}{\Obj}$ the set of such winning vertices $v_0$. 
Similarly we denote by $\Win{1}{G}{\overline\Obj}$ the set of vertices from which player~$1$ can ensure his objective $\overline{\Obj}$. Thus given a player~$i$ and an objective $\Obj$, $\Win{i}{G}{\Obj}$ is the set of vertices from which player~$i$ can ensure $\Obj$ in the game structure~$G$.

A game $(G,\Obj)$ is \emph{determined} if each of its vertices belongs to 
$\Win{0}{G}{\Obj}$ or $\Win{1}{G}{\overline\Obj}$. Martin's theorem~\cite{Martin75} states that all games with Borel objectives are determined. The problem of \emph{solving a game} $(G,\Obj)$ means to decide, given an initial vertex $v_0$, whether player~$0$ is winning from $v_0$ for $\Obj$ (or dually whether player~$1$ is winning from $v_0$ for $\overline\Obj$ when the game is determined). The sets  $\Win{0}{G}{\Obj}$ and $\Win{1}{G}{\overline\Obj}$ are also called the \emph{solutions} of the game.  

\paragraph{\bf Parity and generalized parity objectives.~} 

Let $G$ be a game structure and $d \in \N$ be an integer. Let $\p \colon V \rightarrow [d]$, with $[d] = \{0, 1, \ldots, d\}$, be a \emph{priority} function that associates a priority with each vertex. The \emph{parity} objective $\Obj = \Par(\p)$ asks that the maximum priority seen infinitely often along a play is even, i.e., $\Par(\p) = \{ \pi \in \Plays(G) \mid \max_{v \in \Occinf(\pi)} \p(v) \mbox{ is even}\}$. Games $(G,\Par(\p))$ are called \emph{parity games}. In those games, player~$1$ has the opposite objective $\overline{\Obj}$ equal to $\{ \pi \in \Plays(G) \mid \max_{v \in \Occinf(\pi)} \p(v) \mbox{ is odd}\}$. We denote $\overline\Obj$ by $\OddPar(\p)$. In the sequel, an \emph{$i$-priority} means an even priority if~$i = 0$ and an odd priority if~$i = 1$. For convenience, we also use notation $\iPar{i}(\p)$ such that $\iPar{0}(\p) = \Par(\p)$ and $\iPar{1}(\p) = \OddPar(\p)$. Notice that objective $\iPar{i}(\p)$ is prefix-independent. 

We now consider $k \geq 1$ priority functions $\p_\ell \colon V \rightarrow [d_\ell]$, $\ell \in \{1,\ldots,k\}$. The \emph{generalized parity} objective $\Obj = \ConjPar(\p_1, \ldots, \p_k)$ is the conjunction of the parity objectives defined by all $\p_\ell$, i.e., $\ConjPar(\p_1, \ldots, \p_k) = \bigcap_{\ell = 1}^k \Par(\p_\ell)$. Thus the opposite objective $\overline{\Obj}$ for player~$1$ is equal to $\DisjPar = \bigcup_{\ell = 1}^k \OddPar(\p_\ell)$. Games $(G,\ConjPar(\p_1, \ldots, \p_k))$ are called \emph{generalized parity games}.

Parity games and generalized parity games are determined because their objectives are $\omega$-regular and thus Borel. The complexity of solving those games is stated in the next theorem, as well as the memory requirements for the winning strategies.

\begin{theorem}
\begin{itemize}
\item Solving parity games is in $\mathsf{UP}\cap\mathsf{co}$-$\mathsf{UP}$ and both players have memoryless winning strategies~\cite{Jurdzinski98}.
\item Solving generalized parity games is $\mathsf{co}$-$\mathsf{NP}$-complete, player~$1$ has memoryless winning strategies, and finite-memory strategies are necessary and sufficient for player~$0$ to win~\cite{ChatterjeeHP07}. 
\end{itemize}
\end{theorem}

\paragraph{\bf Partial solvers.~} 

In this paper, we study \emph{partial solvers} for parity games and generalized parity games, that are algorithms that partially compute their solutions. A partial solver returns two partial sets of winning vertices $Z_0 \subseteq \Win{0}{G}{\Obj}$ and $Z_1 \subseteq \Win{1}{G}{\overline\Obj}$ such that $\rest{G}{U}$ is a subgame structure with $U = V \setminus (Z_0 \cup Z_1)$. In the next sections, we present the polynomial time partial solvers proposed in~\cite{HuthKP13,HuthKP16} for parity games and show how to extend them to generalized parity games.

\paragraph{\bf Other $\omega$-regular objectives.~} 

We recall some other useful $\omega$-regular objectives. Given a game structure $G$ and subsets $U, U_1, \ldots, U_k \subseteq V$:
\begin{itemize}
\item the \emph{reachability objective} asks to visit $U$ (called \emph{target set}) at least once, i.e. 
$\Reach(U) = \{ \pi \in \Plays(G) \mid \Occ(\pi) \cap U \ne \emptyset \}$,
\item the \emph{safety objective} asks to avoid visiting $U$ (called \emph{bad set}), i.e. 
$\Safe(U) = \{ \pi \in \Plays(G) \mid \Occ(\pi) \cap U = \emptyset\}$,
\item the \emph{B\"uchi objective} asks to visit infinitely often a vertex of $U$, i.e.
$\Buchi(U) = \{ \pi \in \Plays(G) \mid \Occinf(\pi) \cap U \ne \emptyset \}$,
\item the \emph{co-B\"uchi objective} asks to avoid visiting infinitely often $U$, i.e.
$\CoBuchi(U) = \{\pi \in \Plays(G) \mid \Occinf(\pi) \cap U = \emptyset \}$,
\item the \emph{generalized B\"uchi objective} $\GenBuchi(U_1,\ldots,U_k)$ is equal to the intersection $\bigcap_{\ell = 1}^k \Buchi(U_\ell)$.
\end{itemize}

The next theorem summarizes the time complexities for solving those games as implemented in our prototype tool.\footnote{A better algorithm in $O(|V|^2)$ for B\"uchi objectives is proposed in~\cite{ChatterjeeH14}, and in $O(k\cdot |V|^2)$ for generalized B\"uchi objectives in~\cite{ChatterjeeDHL16}.}

\begin{theorem}\label{thm:complexity} 
For solving games $(G,\Obj)$, we have the following time complexities.
\begin{itemize}
\item Reachability, safety objectives:  $O(|E|)$~\cite{2001automata}.
\item B\"uchi, co-B\"uchi, B\"uchi $\cap$ safety objectives: $O(|V|\cdot|E|)$~\cite{2001automata}.
\item Generalized B\"uchi, generalized B\"uchi $\cap$ safety objectives: $O(k\cdot |V|\cdot|E|)$\footnote{This result is obtained thanks to a classical reduction to games with B\"uchi objectives~\cite{BloemCGHJ10}.}.
\end{itemize}
\end{theorem}

\paragraph{\bf Attractors.~} 
We conclude this section with some basic notions related to the concept of attractor. Let $G$ be a game structure. The \emph{controllable predecessors} for player~$i$ of a set $U \subseteq V$, denoted by $\CPre{i}{G}{U}$, is the set of vertices from which player~$i$ can ensure to visit $U$ in \emph{one step}. Formally, 
\begin{eqnarray}
\CPre{i}{G}{U} &=& \{v \in V_i \mid \exists (v,v') \in E, v' \in U \} \cup \nonumber \\
&& \{v \in V_{1-i} \mid \forall (v,v') \in E, v' \in U \}. \label{eq:Cpre}
\end{eqnarray}
The \emph{attractor $\Attr{i}{G}{U}$ for player~$i$} is the set of vertices from which he can ensure to visit $U$ in \emph{any} number of steps (including zero steps). It is constructed by induction as follows: $\Attr{i}{G}{U} =  \bigcup_{j \geq 0} X_j$ such that:
\begin{eqnarray*}
X_0 &=& U, \\
X_{j+1} &=& X_j \cup \CPre{i}{G}{X_j} \mbox{ for all } j \in \N. 
\end{eqnarray*}
It is therefore equivalent to the winning set $\Win{i}{G}{\Reach(U)}$. The \emph{\posAttr} $\PosAttr{i}{G}{U}$ is the set of vertices from which player $i$ can ensure to visit $U$ in any \emph{positive} number of steps, that is,  $\PosAttr{i}{G}{U} = \bigcup_{j \geq 0} X_j$ such that:
\begin{eqnarray}
X_0 &=& \CPre{i}{G}{U}, \label{eq:AttrInit} \\
X_{j+1} &=& X_j \cup \CPre{i}{G}{X_j \cup U} \mbox{ for all } j \in \N. \label{eq:Attr} 
\end{eqnarray}

Given $U \subseteq V$, we say that $U$ is an \emph{$i$-trap} if for all $v \in U \cap V_i$ and all $(v,v') \in E$, we have $v' \in U$ (player $i$ cannot leave $U$), and for all $v \in U \cap V_{1-i}$, there exists $(v,v') \in E$ such that $v' \in U$ (player $1-i$ can ensure to stay in $U$). Therefore $\rest{G}{U}$ is a subgame structure. When $V \setminus U$ is an $i$-trap, we also use the notation $\subgame{G}{U}$ (instead of $\rest{G}{V \setminus U}$) for the subgame structure induced by $V \setminus U$. The next properties are classical.

\begin{theorem}[\cite{2001automata}] \label{thm:attr} 
Let $G$ be a game structure, $i \in \{0,1\}$ be a player, and $U \subseteq V$ be a subset of vertices. Then
\begin{itemize}
\item the attractor $\Attr{i}{G}{U}$ and the \posAttr\ $\PosAttr{i}{G}{U}$ can be computed in $O(|E|)$ time,
\item the set $V \ssetminus \Attr{i}{G}{U}$ is an $i$-trap.
\end{itemize}
\end {theorem} 

\section{Zielonka's algorithm with partial solvers} \label{sec:Ziel}




The classical algorithm used to solve parity games is the recursive algorithm proposed by Zielonka in~\cite{zielonka98}. Despite its relatively bad theoretical $O(|V|^d)$ time complexity, it is known to outperform other algorithms in practice~\cite{FriedmannL09,Dijk18}. This algorithm solves parity games $(G,\Par(\p))$ by working in a divide-and-conquer manner, combining solutions of subgames to obtain the solution of the whole game. It returns two sets $\{W_0,W_1\}$ such that $W_i = \Win{i}{G}{\iPar{i}(\p)}$ is the winning set for player~$i$. The recursion in this algorithm is performed both on the number of priorities and of nodes in the game. See Algorithm~\ref{algo:PartialZiel} in which no call to a partial solver is performed (therefore line 4 is to be replaced by $\{\Zset_0,\Zset_1\} = \{\emptyset, \emptyset\}$). This algorithm is called Algorithm \Ziel.

Let us explain how Zielonka's algorithm can be combined with a partial solver for parity games (see Algorithm~\ref{algo:PartialZiel}). When $V$ is not empty, we first execute the partial solver. If it solves completely the game, we are done. Otherwise, let $\overline G$ be the subgame of $G$ that was not solved. We then execute the Zielonka instructions on $\overline{G}$ and we return the union of the partial solutions obtained by the partial solver with the solutions obtained for $\overline G$. Proposition~\ref{prop:PartialZiel} below guarantees the soundness of this approach under the hypothesis that if some player wants to escape from $\overline G$, then he will necessarily go to the partial solution of the other player.

\begin{algorithm}
	\caption{\PartialZiel($G,\p$)}
	\begin{algorithmic}[1]
	\If {$V = \emptyset$} \label{line:basic1}
		\State $W_0 = \emptyset$, $W_1 = \emptyset$
		\State return $\{W_0,W_1\}$
	\Else
		\State $\{\Zset_0,\Zset_1\} =$ PSolver$(G,\alpha)$ \label{line:partialsol}
		\State $\overline{G} = G \setminus (\Zset_0 \cup \Zset_1)$, $\overline{V} = V \setminus (\Zset_0 \cup \Zset_1)$
		 	\If  {$\overline{V} = \emptyset$} \label{line:basic2}
			\State return $\{\Zset_0,\Zset_1\}$
		\Else
			\State $p = \max \{\p(v) \mid v \in \overline{V}\}$
			\State $i = p \bmod 2$
			\State $U = \{v \in \overline{V} \mid \p(v) = p\}$
			\State $X = \Attr{i}{\overline{G}}{U}$ \label{line:X}
			\State $\{W'_i,W'_{1-i}\} =$ \PartialZiel$(\overline{G} \setminus X,\p)$
			\If {$W'_{1-i} = \emptyset$}
				\State $W_{i} = \Zset_i \cup W'_i \cup X$, $W_{1-i} = \Zset_{1-i}$
			\Else 
				\State $X = \Attr{1-i}{\overline{G}}{W'_{1-i}}$ \label{line:Y}
				\State $\{W''_i,W''_{1-i}\} =$ \PartialZiel$(\overline{G} \setminus X,\p)$
				\State $W_{i} = \Zset_i \cup W''_{i}$, $W_{1-i} = \Zset_{1-i} \cup W''_{1-i} \cup X$
			\EndIf
			\State return $\{W_i,W_{1-i}\}$ 
		\EndIf
	\EndIf
	\end{algorithmic}
	\label{algo:PartialZiel}
\end{algorithm}


\begin{proposition} \label{prop:PartialZiel}
Suppose that the partial solver used in Algorithm~\ref{algo:PartialZiel} computes partial solutions $\Zset_0,\Zset_1$ such that for all $(v,v') \in E$ and $i \in \{0,1\}$, if $v \in \overline{V} \cap V_i$ and $v' \not\in \overline{V}$, then $v' \in Z_{1-i}$.
Then Algorithm \PartialZiel\ correctly computes the sets $\Win{i}{G}{\iPar{i}(\p)}$, for $i \in \{0,1\}$.
\end{proposition}

\begin{proof}
The proof is by induction and it supposes the soundness of both the (classical) \Ziel\ algorithm and the partial solver. Clearly, Algorithm~\ref{algo:PartialZiel} is correct in the basic cases $V = \emptyset$ (line~\ref{line:basic1}) and $V \setminus Z = \emptyset$ (line~\ref{line:basic2}). 
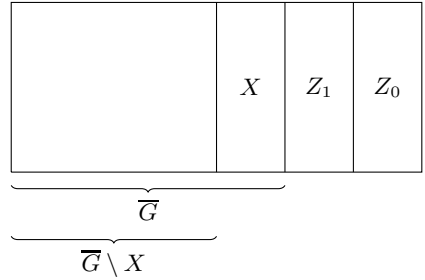
\begin{figure}[ht]
	\centering
	\begin{tikzpicture}[scale=0.9]
		\draw (0,0) -- (6,0) -- (6,2.5) -- (0,2.5) -- (0,0);
		\draw (3,0) -- (3,2.5);
		\draw (4,0) -- (4,2.5);
		\draw (5,0) -- (5,2.5);
		\node[] at (3.5,1.25) {$X$};
		\node[] at (4.5,1.25) {$Z_1$};
		\node[] at (5.5,1.25) {$Z_0$};
		
		\draw[decoration={brace,mirror,raise=5pt},decorate] (0,0) -- node[below=7pt] {$\overline{G}$} (4,0);
		\draw[decoration={brace,mirror,raise=5pt},decorate] (0,-0.75) -- node[below=7pt] {$\overline{G}\setminus X$} (3,-0.75);
	\end{tikzpicture}
	\caption{The partial solutions~$Z_0, Z_1, X$ before a recursive call to Algorithm~\ref{algo:PartialZiel}.}
	\label{fig:PartialZiel}
\end{figure}

Otherwise (see Figure~\ref{fig:PartialZiel}), let $Z_0, Z_1$ be the partial solutions computed by the partial solver (line~\ref{line:partialsol}). As $\overline G$ is a subgame by definition of a partial solver, a recursive call to Algorithm~\ref{algo:PartialZiel} can be executed with ${\overline{G} \setminus X}$ such that $X$ is the attractor computed in line~\ref{line:X} or~\ref{line:Y}. We know that:
\begin{enumerate}
    \item the partial solutions computed by the partial solver are correct,
    \item by induction hypothesis, the solutions computed recursively on $\overline{G} \setminus X$ are correct,
    \item by Algorithm \Ziel, the solutions computed on $\overline{G}$ are thus correct, such that $X$ is part of the solution for some player~$i_0$,
    \item by Theorem~\ref{thm:attr}, as $X$ is an attractor for player~$i_0$, $\overline{V}\setminus X$ is a $i_0$-trap in $\overline G$.   
\end{enumerate}
We can thus define strategies $\sigma_i$, $i \in \{0,1\}$, in the whole game $G$ that are built (in the expected way) from the winning strategies of both players in the subgame $\overline G$ (by items 2 and 3) and their winning strategies with respect to $Z_0, Z_1$ in $G$ (by item 1). These strategies are winning because: 
\begin{itemize}
    \item any player~$i$ that can escape from $\overline G$ has no interest to escape because he will go to $Z_{1-i}$ by the hypothesis of Proposition~\ref{prop:PartialZiel}, a part of $G$ where he is not winning (by item 1),
    \item only player~$1 - i_0$ can escape from $\overline{G} \setminus X$ to go to $X$ (by item 4), but he has no interest to escape since he is not winning in $X$ (by item 3).\qed
\end{itemize}
\end{proof}

An extension of Zielonka's algorithm to generalized parity games is introduced in \cite{ChatterjeeHP07}\footnote{This algorithm is referred to as ``the classical algorithm'' in~\cite{ChatterjeeHP07}.}. This algorithm, that we call \GenZiel, has $O(|E|\cdot |V|^{2D})\binom{D}{d_1,\ldots,d_k}$ time complexity where $D = \Sigma_{\ell = 1}^k d_\ell$. While more complex, it has the same behavior with respect to the recursive call: an attractor $X$ is computed as part of the solution of one player and a recursive call is executed on the subgame $G \setminus X$. Therefore, as done in Algorithm~\ref{algo:PartialZiel}, this algorithm can be combined with a partial solver for generalized parity games, as long as the latter satisfies the hypothesis of Proposition~\ref{prop:PartialZiel}.

In the following sections, we present three partial solvers for parity games and their extension to generalized parity games. They all satisfy the assumptions of Proposition~\ref{prop:PartialZiel} because the partial solutions that they compute, if non empty, are composed of one or several attractors (see Theorem~\ref{thm:attr}).

\section{Algorithms \psolB\ and \GenpsolB} \label{sec:psolB}

In this section, we present simple partial solvers for parity games and generalized parity games. More elaborate and powerful partial solvers are presented in the next two sections. These first partial solvers are based on Propositions~\ref{prop:psolB} and~\ref{prop:GenpsolB} that are direct consequences of the definition of parity games and generalized parity games.
The first proposition states that for parity games, if player~$i$ can ensure to visit infinitely often an $i$-priority without visiting infinitely often a greater $(1-i)$-priority, then he is winning for $\iPar{i}(\p)$.

\begin{proposition} \label{prop:psolB}
Let $(G,\Par(\p))$ be a parity game and $p \in [d]$ be an $i$-priority. Let $U = \{v \in V \mid \p(v) = p \}$ and $U' = \{ v \in V \mid \p(v)$ is a $(1-i)$-priority and $\p(v) > p \}$. For all $v_0 \in V$, if $v_0 \in \Win{i}{G}{\Buchi(U) \cap \CoBuchi(U')}$, then $v_0 \in \Win{i}{G}{\iPar{i}(\p)}$.
\qed\end{proposition}

The second proposition states that for generalized parity games, \emph{(i)} if player~$0$ can ensure to visit infinitely often a $0$-priority $p_\ell$ without visiting infinitely often a $1$-priority greater than $p_\ell$  on \emph{all dimensions} $\ell$, then he is winning in the generalized parity game, and \emph{(ii)} if player~$1$ can ensure to visit infinitely often a $1$-priority $p_\ell$ without visiting infinitely often a $0$-priority greater than $p_\ell$  on \emph{some dimension} $\ell$, then he is winning in the generalized parity game. 

\begin{proposition} \label{prop:GenpsolB}
Let $(G,\ConjPar(\p_1,\ldots,\p_k))$ be a generalized parity game.
\begin{itemize}
\item Let $p = (p_1,\ldots,p_k)$, with $p_\ell \in [d_\ell]$, be a vector of $0$-priorities. For all $\ell$, let $U_\ell = \{v \in V \mid \p_\ell(v) = p_\ell \}$. Let $U' = \{ v \in V \mid \exists \ell, \, \p_\ell(v)$ is a $1$-priority and $\p_\ell(v) > p_\ell \}$. For all $v_0 \in V$, 
\begin{align*}
     & v_0 \in \Win{0}{G}{\GenBuchi(U_1,\ldots,U_k) \cap \CoBuchi(U')}\\
    \implies & v_0  \in \Win{0}{G}{\ConjPar(\p_1,\ldots,\p_k)}
\end{align*}
\item Let $\ell \in \{1,\ldots, k\}$ and $p_\ell \in [d_\ell]$ be a $1$-priority. Let $U = \{v \in V \mid \p_\ell(v) = p_\ell \}$ and $U' = \{ v \in V \mid \p_\ell(v)$ is a $0$-priority and $\p_\ell(v) > p_\ell \}$. For all $v_0 \in V$, 
\begin{align*}
     & v_0 \in \Win{1}{G}{\Buchi(U) \cap \CoBuchi(U')}\\
    \implies & \Win{1}{G}{\DisjPar(\p_1,\ldots,\p_k)}
\end{align*}
\end{itemize}
\qed\end{proposition}

\paragraph{\bf Partial Solver for Parity Games.~}
A partial solver for parity games can be easily derived from Proposition~\ref{prop:psolB}.  Given a parity game $(G,\Par(\p))$, the polynomial time algorithm~\psolB\ (see Algorithm~\ref{algo:psolB}) tries to find winning vertices for each player by applying this proposition as follows. Let us denote by $W_p$ the winning set computed by the algorithm for priority $p$ (line~\ref{line:win}) and let us suppose (for simplicity) that the loop in line~\ref{line:loop} treats the priorities from the highest one $d$ to the lowest one $0$. This algorithm thus computes $W_d$, $W_{d-1}, \ldots$, until finding $W_p \neq \emptyset$. At this stage, the algorithm was able to find some winning vertices. It then repeats the process on the subgame $\subgame{G}{W_p}$ to find other winning vertices. At the end of the execution, it returns the partial solutions $\{Z_0,Z_1\}$ that it was able to compute. 

\begin{algorithm}
	\caption{\psolB($G,\p$)}
	\begin{algorithmic}[1]
		\ForEach{$p \in [d]$} \label{line:loop}
		    \State $i = p \bmod2$
			\State $U = \{v \in V \mid \p(v) = p\}$
            \State $U' = \{ v \in V \mid \p(v) \mbox{ is a $(1-i)$-priority and } \p(v) > p \}$
            \State $W = \Win{i}{G}{\Buchi(U) \cap \CoBuchi(U')}$ \label{line:win}
			\If {$W \neq \emptyset$} 
                \State $\{Z_i,Z_{1-i}\} =$ \psolB($G \setminus W,\p)$ \label{line:subgame}
			    \State return $\{Z_i \cup W,Z_{1-i}\}$ 
			\EndIf
		\EndFor
		\State return $\{\emptyset,\emptyset\}$ 
	\end{algorithmic}
	\label{algo:psolB}
\end{algorithm}

Notice that we could replace line~\ref{line:win} by 
$$W' = \Win{i}{G}{\Buchi(U) \cap \Safe(U')},~~  W = \Attr{i}{G}{W'}$$
since $\Attr{i}{G}{W'} \subseteq \Win{i}{G}{\Buchi(U) \cap \CoBuchi(U')}$ (as the parity objective is prefix-independent, it is closed under attractor; and computing the attractor is necessary to get a subgame for the recursive call).   
This variant is investigated in~\cite{HuthKP13,HuthKP16} under the name of Algorithm~psolB.


\paragraph{\bf Partial Solver for Generalized Parity Games.~}
Similarly, a partial solver for generalized parity games can be derived from Proposition~\ref{prop:GenpsolB}. Instead of considering all $p \in [d]$ as done in line~\ref{line:loop} of Algorithm~\psolB, we have here a more complex loop on all elements that are
\begin{itemize}
\item either a \emph{$1$-priority} $p$ for some given \emph{dimension} $\ell$ (case of player~$1$),
\item or a \emph{vector} $(p_1,\ldots,p_k)$ with \emph{$0$-priorities} $p_\ell$ (case of player~$0$).
\end{itemize}
We suppose that we have a \emph{list $\List$} of all possible such elements $(p,\ell)$ and $(p_1,\ldots,p_k)$. Given a generalized parity game, Algorithm~\GenpsolB\ (see Algorithm~\ref{algo:GenpsolB}) tries to find winning vertices for both players by applying Proposition~\ref{prop:GenpsolB}. It returns the partial solutions $\{Z_0,Z_1\}$ that it was able to compute. 
As in Algorithm~\psolB, objective $\CoBuchi(U')$ can be replaced by $\Safe(U')$ in lines~\ref{line:Safe2} and~\ref{line:Safe1} (with the addition of an attractor computation). This modification yields an algorithm in $O((\Pi_{\ell = 1}^k  \frac{d_\ell}{2}) \cdot k \cdot |V|^2 \cdot |E|)$ time by Theorem~\ref{thm:complexity}.

\begin{algorithm}
	\caption{\GenpsolB($G,\p_1,\ldots,\p_k, \List$)}
	\begin{algorithmic}[1]
		\ForEach{$element \in \List$}
		\If {$element = (p,\ell)$}
				    \State $U = \{v \in V \mid \p_\ell(v) = p\}$ \label{line:begin}
            \State $U' = \{ v \in V \mid \p_\ell(v)  \mbox{ is a $0$-priority and } \p_\ell(v) > p \}$
            \State $W = \Win{1}{G}{\Buchi(U) \cap \CoBuchi(U')}$ \label{line:Safe2} \label{line:end}
			\If {$W \neq \emptyset$} 
		        \State $\{Z_0,Z_1\} =$ \GenpsolB($G \setminus W,\p_1,\ldots,\p_k,\List$)
			    \State return $\{Z_0,Z_1 \cup W\}$
		    \EndIf
		\Else \quad\quad\{We know that $element = (p_1,\ldots,p_k)$\}
		    \ForEach{$\ell$}
			    \State $U_\ell = \{v \in V \mid \p_\ell(v) = p_\ell \}$
	        \EndFor
            \State $U' = \{ v \in V \mid \exists \ell, \, \p_\ell(v) \mbox{ is a $1$-priority and } \p_\ell(v) > p_\ell \}$
            \State $W = \Win{0}{G}{\GenBuchi(U_1,\ldots,U_k) \cap \CoBuchi(U')}$ \label{line:Safe1}
		    \If {$W \neq \emptyset$}
		        \State $\{Z_0,Z_1\} =$ \GenpsolB($G \setminus W,\p_1,\ldots,\p_k,\List$)
			    \State return $\{Z_0 \cup W,Z_1\}$
		    \EndIf
		\EndIf
		\EndFor
		\State return $\{\emptyset,\emptyset\}$ 
	\end{algorithmic}
	\label{algo:GenpsolB}
\end{algorithm}



\section{Algorithms \psolC\ and \GenpsolC}  \label{sec:psolC}

In this section, we present a second polynomial time partial solver for parity games as proposed in~\cite{HuthKP16}. We then explain how to extend it for partially solving generalized parity games. We finally explain how the latter solver can be transformed into a more efficient (in practice) antichain-based algorithm. 

\paragraph{\bf Partial solver for parity games.~}

In Section~\ref{sec:psolB}, given a parity game, the basic action of Algorithm~\psolB\ is to compute the set of winning vertices for player~$i$ for the objective $\Buchi(U) \cap \CoBuchi(U')$ that imposes to visit infinitely often a given $i$-priority $p$ (set $U$) without visiting infinitely often any greater $(1-i)$-priority (set $U'$). In this section, this approach is extended to the \emph{set of all $i$-priorities} instead of only one $i$-priority $p$.

To this end, we consider the \emph{extended} game structure $G \times M$ with $M = [d]$  such that $m \in M$ records the \emph{maximum visited priority}. More precisely, the set of vertices of $G \times M$ is equal to $V \times M$ where $V_i \times M$ is the set of vertices controlled by player~$i$, and the set $\edge{M}$ of its edges is composed of all pairs $((v,m),(v', m'))$ such that $(v,v') \in E$ and $m' = \max\{m,\p(v)\}$. Clearly, with this construction, in $G$, player~$i$ can ensure to visit $v$ from $v_0$ such that the maximum visited priority ($v$ excluded) is an $i$-priority if and only if in $G \times M$, he can ensure to visit $(v,m)$ from $(v,\p(v_0))$ for some $i$-priority $m$. 

We want to compute a set $\JF{i}{G}{\p}$ such that if $v_0 \in \JF{i}{G}{\p}$, then player~$i$ is winning from $v_0$ for $\iPar{i}(\p)$. We proceed by computing the following sequence $F^{(i)} = \bigcap_{j \geq 0}F_j$. Initially $F_0 = V$ and for $j \geq 1$, $F_j$ is computed from $F_{j-1}$ as follows:
\begin{eqnarray}
T_j &=& \{ (v,m) \in V \times M \mid v \in F_{j-1} \mbox{ and $m$ is an $i$-priority} \} \label{eq:begin} \\
A_j &=& \PosAttr{i}{G \times M}{T_j} \label{eq:mid} \\
F_j &=& \{ v \in V \mid (v,\init(v)) \in A_j)\} \cap F_{j-1}.  \label{eq:end}
\end{eqnarray}
Intuitively, if $v_0 \in F_j$, then player~$i$ has a strategy to ensure to visit some vertex $v \in F_{j-1}$ such that the maximum visited priority along the consistent history $hv$ from $v_0$ to $v$ (priority of $v$ excluded) is some $i$-priority. We say that $h$ is a \emph{good episode}. Notice that $h$ is non empty since each $A_j$ is a \posAttr. 

We define $\JF{i}{G}{\p}$ as the fixpoint $F^{(i)}$. Therefore from a vertex $v_0$ in this fixpoint, player~$i$ can ensure a succession of good episodes in which the maximum visited priority is an $i$-priority. Thus he is winning from $v_0$ for $\iPar{i}(\p)$ as formalized in the next proposition. Notice that if $v_0$ belongs to the attractor for player~$i$ of $\JF{i}{G}{\p}$, then $v_0$ still belongs to $\Win{i}{G}{\iPar{i}(\p)}$ as the latter objective is prefix-independent.

\begin{proposition}[\cite{HuthKP16}] \label{prop:JFpsolC}
Let $(G,\Par(\p))$ be a parity game. Then for all $v_0 \in V$, if $v_0 \in \Attr{i}{G}{F^{(i)}}$, then $v_0 \in \Win{i}{G}{\iPar{i}(\p)}$. 
\end{proposition}

From Proposition~\ref{prop:JFpsolC}, we derive the polynomial time algorithm called~\psolC\ (see Algorithm~\ref{algo:psolC}). This algorithm is called psolC in~\cite{HuthKP16}. 

\begin{algorithm}
	\caption{\psolC($G,\p$)}
	\begin{algorithmic}[1]
		\ForEach{$i \in \{0,1\}$}
                         \State $W = \JF{i}{G}{\p}$
			\If {$W \neq \emptyset$} 
			    \State $X = \Attr{i}{G}{W}$
			    \State $\{Z_i,Z_{1-i}\} =$ \psolC($G \setminus X,\p$) 
			    \State return $\{Z_i \cup X,Z_{1-i}\}$
			\EndIf
		\EndFor
		\State return $\{\emptyset,\emptyset\}$ 
	\end{algorithmic}
	\label{algo:psolC}
\end{algorithm}


\paragraph{\bf Partial solver for generalized parity games.~}

The same approach can be applied to generalized parity games with some adaptations that depend on the player. Recall that the basic action of Algorithm~\GenpsolB\ in Section~\ref{sec:psolB} is \emph{(i)} to test whether player~$0$ can ensure to visit infinitely often a $0$-priority $p_\ell$ without visiting infinitely often a $1$-priority greater than $p_\ell$  on all dimensions $\ell$, or \emph{(ii)} to test whether player~$1$ can ensure to visit infinitely often a $1$-priority $p_\ell$ without visiting infinitely often a $0$-priority greater than $p_\ell$ on some dimension~$\ell$.

Therefore, for player~$1$, in place of lines~\ref{line:begin}-\ref{line:end} of Algorithm~\ref{algo:GenpsolB} we can now apply the \psolC\ approach by computing $W = \JF{1}{G}{\p_\ell}$.

For player~$0$, we have to deal with vectors of $0$-priorities. We thus consider the extended game structure $G \times M_1 \times \ldots \times M_k$ such that for all $\ell$, $M_\ell$ is equal to $[d_\ell]$, and where $m_\ell \in M_\ell$ records the maximum visited priority in \emph{dimension $\ell$}. Hence, the edges of this game structure have the form $((v,m_1, \ldots,m_k),(v', m'_1,\ldots,m'_k))$ such that $(v,v') \in E$ and $m'_\ell = \max\{m_\ell,\p_\ell(v)\}$ for all $\ell$.  
A good episode is now a history $h$ such that for all $\ell$, the maximum priority visited along $h$ in dimension $\ell$ is a $0$-priority. And the related set $\JF{0}{G}{\p_1,\ldots,\p_\ell}$ equal to $\bigcap_{j\geq 0} F_j$ is computed with Equations~(\ref{eq:begin}-\ref{eq:end}) modified as expected. 
We have the next counterpart of Proposition~\ref{prop:JFpsolC} with a similar proof.
 
\begin{proposition} \label{prop:GenJFpsolC}
Let $(G,\ConjPar(\p_1,\ldots,\p_k))$ be a generalized parity game. If $v_0 \in \Attr{0}{G}{F^{(i)}}$, then $v_0 \in \Win{0}{G}{\ConjPar(\p_1,\ldots,\p_k)}$.
\qed\end{proposition}

For all the previous arguments, we derive Algorithm~\GenpsolC\ (see Algorithm~\ref{algo:GenpsolC}). It tries to find a partial solution first for player~$1$ (for some dimension~$\ell$) and then for player~$2$. By Theorem~\ref{thm:complexity}, its time complexity is in $O((\Pi_{\ell = 1}^{k} d_\ell) \cdot|V|^2\cdot|E|)$.

\begin{algorithm}
	\caption{\GenpsolC($G,\p_1,\ldots,\p_k$)}
	\begin{algorithmic}[1]
		\ForEach{$\ell \in \{1,\ldots,k\}$}
			\State $W = \JF{1}{G}{\p_\ell}$  \label{line:JF2}
			\If {$W \neq \emptyset$} 
			    \State $X = \Attr{1}{G}{W}$
		        \State $\{Z_0,Z_1\} =$  \GenpsolC($G \setminus X,\p_1,\ldots,\p_k$)
				\State return $\{Z_0,Z_1 \cup X\}$	        
		    \EndIf
		 \EndFor
	     \State $W = \JF{0}{G}{\p_1,\ldots,\p_k}$ \label{line:JF1}
		 \If {$W \neq \emptyset$} 
		        \State $X = \Attr{0}{G}{W}$
		        \State $\{Z_0,Z_1\} =$  \GenpsolC($G \setminus X,\p_1,\ldots,\p_k$)
			    \State return $\{Z_0 \cup X,Z_1\}$
		 \EndIf
		 \State return $\{\emptyset,\emptyset\}$ 
	\end{algorithmic}
	\label{algo:GenpsolC}
\end{algorithm}

%
%
%
%
%
%

\paragraph{\bf Partial solvers with antichains.~}

Algorithm~\GenpsolC\ has exponential time complexity due to the use of the extended structure $G \times M_1 \times \ldots \times M_k$ associated to a generalized parity game, and in particular to the computation of $\JF{0}{G}{\p_1,\ldots,\p_k}$. We here show that the vertices of this extended game can be partially ordered in a way to obtain an antichain-based algorithm for the computation of $\JF{0}{G}{\p_1,\ldots,\p_k}$. The antichains allow compact representation and efficient manipulation of partially ordered sets \cite{DoyenR10}.

Let us first recall the basic notions about antichains. Consider a \emph{partially ordered set} $(S,\preceq)$ where $S$ is a finite set and $\preceq \, \subseteq S \times S$ is a partial order on $S$. Given $s, s'\in S$, we write $s \sqcap s'$ their \emph{greatest lower bound} if it exists. 
A \emph{lower semilattice} is a partially ordered set such that this greater lower bound always exists for all $s,s' \in S$. Given two subsets $R, R' \subseteq S$, we denote by $R \sqcap R'$ the set $\{s \sqcap s' \mid s \in R, s' \in R'\}$. 

Given a lower semilattice $(S,\preceq)$, an \emph{antichain} $A$ is a subset of $S$ composed of pairwise incomparable elements with respect to $\preceq$. Given a subset $R \subseteq S$, we denote $\lceil R \rceil$ the set of its \emph{maximal} elements (which is thus an antichain). We say that $R$ is \emph{closed} if whenever $s \in R$ and $s' \preceq s$, then $s' \in R$. If $A$ is an antichain, we denote by $\downarrow \! A$ the closed set that it \emph{represents}, that is, $A = \lceil \downarrow \! A \rceil$. Hence when $R$ is closed, we have $R = \, \downarrow \! \lceil R \rceil$. The benefit of antichains is that they provide a \emph{compact representation} of closed sets. Moreover some operations on those closed sets can be done at the level of their antichains as indicated in the next proposition. 

\begin{proposition}[\cite{DoyenR10}] \label{prop:antichain}
Let $(S,\preceq)$ be a lower semilattice, and $R, R' \subseteq S$ be two closed sets represented by their antichains $A = \lceil R \rceil, A' = \lceil R' \rceil$. Then
\begin{itemize}
\item for all $s \in S$, $s \in R$ if and only if there exists $s' \in A$ such that $s \preceq s'$,
\item $R \cup R'$ is closed, and $\lceil R \cup R' \rceil = \lceil A \cup A' \rceil$,
\item $R \cap R'$ is closed, and $\lceil R \cap R' \rceil = \lceil A \sqcap A' \rceil$.
\end{itemize}
\end{proposition}

For \emph{simplicity}, we focus on the extended structure $G \times M$ associated to a parity game and begin to explain an antichain-based algorithm for the computation of set $\JF{0}{G}{\p}$ (this algorithm is inspired from~\cite{FiliotJR13}). We will explain later what are the needed adaptations to compute $\JF{0}{G}{\p_1,\ldots,\p_k}$ in generalized parity games. We equip $V \times M$ with the following partial order:


\begin{definition} \label{def:order}
We define the strict partial order $\prec$ on $V \times M$ such that $(v',m') \prec (v,m)$ if and only if $v = v'$ and:
\begin{enumerate}
\item either $m, m'$ are even and $m' > m$,
\item or $m, m'$ are odd  and $m' < m$,
\item or $m$ is odd and $m'$ is even.
\end{enumerate}
We define $(v',m') \preceq (v,m)$ if either $(v',m') = (v,m)$ or $(v',m') \prec (v,m)$.
\end{definition}

For instance, if $[d] = [4]$, then $(v,4) \prec (v,2) \prec (v,0) \prec (v,1) \prec (v,3)$. With this definition, two elements $(v,m),(v',m')$ are incomparable as soon as $v \neq v'$. It follows that in Proposition~\ref{prop:antichain}, if $R \subseteq \{v\} \times M$ and  $R' \subseteq \{v'\} \times M$ with $v \neq v'$, then the union $A \cup A'$ of their antichains is already an antichain, that is, $A \cup A' = \lceil A \cup A' \rceil$. Moreover, the operator $\sqcap$ useful in Proposition~\ref{prop:antichain} is defined as indicated in the next lemma.

\begin{lemma}
Let $(v,m),(v,m') \in V \times M$. Then $(v,m) \sqcap (v,m')$ is equal to:
\begin{itemize}
    \item $(v,\max\{m,m'\})$ if $m,m'$ are even,
    \item $(v,\min\{m,m'\})$ if $m,m'$ are odd,
    \item $(v,m)$ if $m$ is even and $m'$ is odd,
    \item $(v,m')$ if $m$ is odd and $m'$ is even.
\end{itemize}
\qed\end{lemma}

The computation of $\JF{0}{G}{\p}$ is based on Equations~(\ref{eq:begin}-\ref{eq:end}). Equation~(\ref{eq:mid}) involves the computation of \posAttr s over $G \times M$ thanks to Equations~(\ref{eq:Cpre}-\ref{eq:Attr}). We are going to show that $\JF{0}{G}{\p}$ is a closed set as well as all the intermediate sets to compute it. It will follow that all the computations can be performed directly on the antichains that represent those sets.

We already know from Proposition~\ref{prop:antichain} that the family of closed sets is stable under union and intersection and how the computations can be done at the level of their antichains. So we now focus on the basic operation $\CPre{0}{G \times M}{U}$ form some $U \subseteq V \times M$. We need to introduce the following functions $\up$ and $\down$. Given $(v,m) \in V \times M$, we define $\up(m,\p(v)) = \max \{m,\p(v)\}$. Recall that such an update from $m$ to $m' = \up(m,\p(v))$ is used in the edges $((v,m),(v',m'))$ of $G \times M$. Function $\down$ is defined hereafter. It is the inverse reasoning of function $\up$: given $(v',m') \in V \times M$ and $(v,v') \in E$, the value $\down(m',\alpha(v))$ yields the maximal value $m$ such that $(v',\up(m,\p(v))) \preceq (v',m')$.

\begin{definition} \label{def:down}
Given $(v',m') \in V \times M$ and $p = \p(v)$, we define $m = \down(m',p)$ as follows:


\begin{enumerate}
    \item Case $p$ even:  
    \begin{itemize}
        \item if $p < m'$ then $m = m'$,
        \item else $m = \max \{p-1,0\}$,
    \end{itemize}
    \item Case $p$ odd: 
    \begin{itemize}
        \item if $p \leq m'$ then $m = m'$, 
        \item else $m = p+1$ except if $p = d$ in which case $\down(m',p)$ is not defined.
    \end{itemize}
\end{enumerate}
\end{definition}





The next proposition states that if the set $U \subseteq V \times M$ is closed, then the set $\CPre{0}{G \times M}{U}$ is also closed.
It also indicates how to design an antichain-based algorithm for computing $\CPre{0}{G \times M}{U}$. For convenience, we recall Equation~(\ref{eq:Cpre}): $\CPre{0}{G \times M}{U} = C_0 \cup C_1$ with
\begin{eqnarray}
C_0 &=& \{(v,m) \in V_0 \times M  \mid \exists ((v,m),(v',m')) \in \edge{M} \mbox{ with } (v',m') \in U \} \label{eq:CpreU0} \\
C_1 &=& \{(v,m) \in V_{1} \times M \mid \forall ((v,m),(v',m')) \in \edge{M}: (v',m') \in U \}. \label{eq:CpreU1}
\end{eqnarray}

\begin{proposition} \label{prop:CPreClosed}
If $U \subseteq V \times M$ is a closed set, then $\CPre{0}{G \times M}{U}$ is closed. Let $A = \lceil U \rceil$ be the antichain representing $U$. Then $\lceil \CPre{0}{G \times M}{U} \rceil = \lceil B_0 \cup B_1 \rceil$ where $B_0, B_1$ are the following antichains:
\begin{eqnarray*}
    B_0  &=& \bigcup_{v \in V_0} \Big\lceil \{ (v,\down(m',\p(v))) \mid (v,v') \in E, (v',m') \in A \} \Big\rceil, \\
    B_1 &=& \bigcup_{v \in V_1} \Big\lceil \bigsqcap_{(v,v')\in E} \big\lceil \{ (v,\down(m',\p(v))) \mid (v',m') \in A \} \big\rceil \Big\rceil. 
\end{eqnarray*}
\end{proposition}

Let us come back to the computation of $\JF{0}{G}{\p}$, which depends on Equations~(\ref{eq:begin}-\ref{eq:end}) and Equations~(\ref{eq:Cpre}-\ref{eq:Attr}) for the \posAttr s of Equation~(\ref{eq:mid}). By the next lemma, all sets $T_j$ of Equation~(\ref{eq:begin}) are closed, and thus, by Propositions~\ref{prop:antichain} and~\ref{prop:CPreClosed}, we get an antichain-based algorithm for computing $\JF{0}{G}{\p}$ as announced. 

\begin{lemma} \label{lem:antichain}
Each $T_j$ of Equation~(\ref{eq:begin}) is a closed set that is represented by the antichain $\lceil T_j \rceil = \{(v,0) \mid v \in F_j\}$.
\qed \end{lemma}

It remains to prove Proposition~\ref{prop:antichain}. We first prove the next properties of functions $\up$ and $\down$.

\begin{lemma} \label{lem:down}
For all $(v,m), (v,m_1), (v,m_2) \in V \times M$ with $(v,v') \in E$ and $p = \p(v)$, 
\begin{enumerate}
    \item if $(v,m_1) \preceq (v,m_2)$, then $(v',\up(m_1,p)) \preceq (v',\up(m_2,p))$,
    \item $(v,m) \preceq (v,\down(\up(m,p),p))$.
\end{enumerate}
For all $(v,m'), (v',m'_1), (v',m'_2) \in V \times M$ with $(v,v') \in E$ and $p = \p(v')$, such that $\down(m',p)$, $\down(m'_1,p)$, $\down(m'_2,p)$ are all defined,
\begin{enumerate}
\setcounter{enumi}{2}
    \item if $(v',m'_1) \preceq (v',m'_2)$, then $(v,\down(m'_1,p)) \preceq (v,\down(m'_2,p))$,
    \item $(v',\up(\down(m',p),p)) \preceq (v',m')$.
\end{enumerate}
\end{lemma}

\begin{proof}
We begin with the proof of Item~1. If $(v,m_1) = (v,m_2)$, then clearly $(v',\up(m_1,p)) = (v',\up(m_2,p))$. Hence we suppose that $(v,m_1) \prec (v,m_2)$.

\begin{enumerate}
\item Suppose that $m_1 < m_2$. According to Definition~\ref{def:order}, this happens when $m_1, m_2$ are odd (Case 2. of Definition~\ref{def:order}) or when $m_1$ is even and $m_2$ is odd (Case 3. of Definition~\ref{def:order}), in summary when \emph{$m_2$ is odd}. If $p < m_1$, then $\up(m_1,p) = m_1$, $\up(m_2,p) = m_2$, and thus $(v',\up(m_1,p)) \prec (v',\up(m_2,p))$. If $m_1 \leq p < m_2$, then $\up(m_1,p) = p$, $\up(m_2,p) = m_2$. It follows that if $p$ is even then $(v',\up(m_1,p)) \prec (v',\up(m_2,p))$ by Case 3. of Definition~\ref{def:order}, and if $p$ is odd then then $(v',\up(m_1,p)) \prec (v',\up(m_2,p))$ by Case 2. of Definition~\ref{def:order}. If $m_2 \leq p$, then $\up(m_1,p) = \up(m_2,p) = p$ and thus $(v',\up(m_1,p)) \preceq (v',\up(m_2,p))$.

\item Suppose that $m_2 < m_1$. According to Definition~\ref{def:order}, this happens when $m_1, m_2$ are even (Case 1. of Definition~\ref{def:order}) or when $m_1$ is even and $m_2$ is odd (Case 3. of Definition~\ref{def:order}), in summary when \emph{$m_1$ is even}. If $p < m_2$, then $\up(m_1,p) = m_1$, $\up(m_2,p) = m_2$, and thus $(v',\up(m_1,p)) \prec (v',\up(m_2,p))$. If $m_2 \leq p < m_1$, then $\up(m_2,p) = p$, $\up(m_1,p) = m_1$. It follows that if $p$ is even then $(v',\up(m_1,p)) \prec (v',\up(m_2,p))$ by Case 1. of Definition~\ref{def:order}, and if $p$ is odd then $(v',\up(m_1,p)) \prec (v',\up(m_2,p))$ by Case 3. of Definition~\ref{def:order}. If $m_1 \leq p$, then $\up(m_2,p) = \up(m_1,p) = p$ and thus $(v',\up(m_1,p)) \preceq (v',\up(m_2,p))$.
\end{enumerate}
%
%
Let us proceed to the proof of Item~2.

\begin{enumerate}
\item Suppose that $p$ is even. If $p<m$, then $\up(m,p) = m$ and thus $\down(\up(m,p),p) = m$ (Case 1. of Definition~\ref{def:down}). It follows that $(v,m) \preceq (v,\down(\up(m,p),p)$. If $p\geq m$ with $p \neq 0$, then $\up(m,p) = p$ and thus $\down(\up(m,p),p)= p-1$. If $m$ is even, then $(v,m) \preceq (v,\down(\up(m,p),p)$ ($p-1$ is odd, Case 3. of Definition~\ref{def:order}); and if $m$ is odd, then $p-1 \geq m$ (since $m,p$ have different parities) and $(v,m) \preceq (v,\down(\up(m,p),p)$ ($p-1$ is odd, Case 2. of Definition~\ref{def:order}). If $p\geq m$ with $p = 0$, then $m=0$, $\up(m,p) = 0$ and thus $\down(\up(m,p),p) = 0$ (Case 1. of Definition~\ref{def:down}). It follows that  $(v,m) \preceq (v,\down(\up(m,p),p)$.

\item Suppose that $p$ is odd. If $p\leq m$, then $\up(m,p) = m$ and thus $\down(\up(m,p),p) = m$ (Case 2. of Definition~\ref{def:down}). It follows that $(v,m) \preceq (v,\down(\up(m,p),p)$. If $p > m$, then $\up(m,p) = p$ and thus $\down(\up(m,p),p) = p$ (Case 2. of Definition~\ref{def:down}). Hence $(v,m) \preceq (v,\down(\up(m,p),p)$ for both possible parities of $m$ (Cases 2 and 3 of Definition~\ref{def:order}). 
\end{enumerate}
We now treat Item~3. By hypothesis, $\down(m'_1,p)$ and $\down(m'_2,p)$ are supposed to be defined. If $(v,m'_1) = (v,m'_2)$, then clearly $(v,\down(m'_1,p)) = (v,\down(m'_2,p))$. Hence we suppose that $(v,m'_1) \prec (v,m'_2)$. 

\begin{enumerate}
\item Suppose that $m'_1 < m'_2$. According to Definition~\ref{def:order}, this happens when $m'_1, m'_2$ are odd (Case 2. of Definition~\ref{def:order}) or when $m'_1$ is even and $m'_2$ is odd (Case 3. of Definition~\ref{def:order}), in summary when \emph{$m'_2$ is odd}. 

Suppose first that $p$ is even (Case 1. of Definition~\ref{def:down}). If $p < m'_1$, then $\down(m'_1,p) = m'_1$, $\down(m'_2,p) = m'_2$, and thus $(v,\down(m'_1,p)) \preceq (v,\down(m'_2,p))$. If $m'_1 \leq p < m'_2$ with $p \neq 0$, then $\down(m'_1,p) = p-1$, $\down(m'_2,p) = m'_2$. It follows that $(v,\down(m'_1,p)) \preceq (v,\down(m'_2,p))$ ($p-1, m'_2$ are odd, Case 2. of Definition~\ref{def:order}). If $m'_1 \leq p < m'_2$ with $p = 0$, then $\down(m'_1,p) = 0$, $\down(m'_2,p) = m'_2$. It follows that $(v,\down(m'_1,p)) \preceq (v,\down(m'_2,p))$ (Case 3. of Definition~\ref{def:order}). If $m'_2 \leq p$, then $\down(m'_1,p) = \down(m'_2,p) = p-1$ and thus $(v,\down(m'_1,p)) \preceq (v,\down(m'_2,p))$. 

Suppose now that $p$ is odd (Case 2. of Definition~\ref{def:down}). If $p \leq m'_1$, then $\down(m'_1,p) = m'_1$, $\down(m'_2,p) = m'_2$, and thus $(v,\down(m'_1,p)) \preceq (v,\down(m'_2,p))$. If $m'_1 < p \leq m'_2$, then $\down(m'_1,p) = p+1$, $\down(m'_2,p) = m'_2$. It follows that $(v,\down(m'_1,p)) \preceq (v,\down(m'_2,p))$ ($p+1$ is even, $m'_2$ is odd, Case 3. of Definition~\ref{def:order}). If $m'_2 < p$, then $\down(m'_1,p) = \down(m'_2,p) = p+1$ and thus $(v,\down(m'_1,p)) \preceq (v,\down(m'_2,p))$. 

\item Suppose that $m'_2 < m'_1$. According to Definition~\ref{def:order}, this happens when $m'_1, m'_2$ are even (Case 1. of Definition~\ref{def:order}) or when $m'_1$ is even and $m'_2$ is odd (Case 3. of Definition~\ref{def:order}), in summary when \emph{$m'_1$ is even}.

Suppose first that $p$ is even (Case 1. of Definition~\ref{def:down}). If $p < m'_2$, then $\down(m'_1,p) = m'_1$, $\down(m'_2,p) = m'_2$, and thus $(v,\down(m'_1,p)) \preceq (v,\down(m'_2,p))$. If $m'_2 \leq p < m'_1$ with $p \neq 0$, then $\down(m'_2,p) = p-1$, $\down(m'_1,p) = m'_1$. It follows that $(v,\down(m'_1,p)) \preceq (v,\down(m'_2,p))$ ($p-1$ is odd, $m'_1$ is even, Case 3. of Definition~\ref{def:order}). If $m'_1 \leq p < m'_2$ with $p = 0$, then $\down(m'_2,p) = 0$, $\down(m'_1,p) = m'_1$. It follows that $(v,\down(m'_1,p)) \preceq (v,\down(m'_2,p))$ (Case 1. of Definition~\ref{def:order}). If $m'_1 \leq p$, then $\down(m'_2,p) = \down(m'_1,p) = p-1$ and thus $(v,\down(m'_1,p)) \preceq (v,\down(m'_2,p))$. 

Suppose now that $p$ is odd (Case 2. of Definition~\ref{def:down}). If $p \leq m'_2$, then $\down(m'_1,p) = m'_1$, $\down(m'_2,p) = m'_2$, and thus $(v,\down(m'_1,p)) \preceq (v,\down(m'_2,p))$. If $m'_2 < p \leq m'_1$, then notice that $p+1 \leq m'_1$ as $p, m'_1$ have opposite parities. Moreover $\down(m'_2,p) = p+1$, $\down(m'_1,p) = m'_1$. It follows that $(v,\down(m'_1,p)) \preceq (v,\down(m'_2,p))$ ($p+1, m'_1$ are even, Case 1. of Definition~\ref{def:order}). If $m'_2 < p$, then $\down(m'_2,p) = \down(m'_1,p) = p+1$ and thus $(v,\down(m'_1,p)) \preceq (v,\down(m'_2,p))$.
\end{enumerate}
It remains to consider the last item. By hypothesis, $\down(m,p)$ is supposed to be defined.
\begin{enumerate}
\item Suppose that $p$ is even (Case 1. of Definition~\ref{def:down}). If $p<m'$, we have that $\down(m',p) = m'$ and thus $\up(\down(m',p),p) = m'$. It follows that $(v',\up(\down(m',p),p) \preceq (v',m')$. If $p\geq m'$ with $p \neq 0$, we have that $\down(m',p) = p-1$ and thus $\up(down(m',p),p) = p$. It follows that if $m'$ is even then $(v',\up(\down(m',p),p) \preceq (v',m')$ by Case 1. of Definition~\ref{def:order}; and if $m'$ is odd then $(v',\up(\down(m',p),p) \preceq (v',m')$ by Case 3. of Definition~\ref{def:order}. If $p\geq m'$ with $p = 0$, then $m'=0$, $\down(m',p) = 0$ and thus $\up(\down(m,p),p) = 0$. It follows that  $(v',\up(\down(m',p),p) \preceq (v',m')$.

\item Suppose that $p$ is odd (Case 2. of Definition~\ref{def:down}). If $p \leq m'$, then we have that $\down(m',p) = m'$ and thus $\up(\down(m',p),p) = m'$. It follows that $(v',\up(\down(m',p),p) \preceq (v',m')$. If $p > m'$, then $\down(m',p) = p+1$ and thus $\up(\down(m',p),p) = p+1$. It follows that if $m'$ is even then $(v',\up(\down(m',p),p) \preceq (v',m')$ by Case 1. of Definition~\ref{def:order}; and if $m'$ is odd then $(v',\up(\down(m',p),p) \preceq (v',m')$ by Case 3. of Definition~\ref{def:order}.
\end{enumerate}
\qed\end{proof}

\begin{proof}[of Proposition~\ref{prop:CPreClosed}]
We first prove that if $U$ is closed, then $\CPre{0}{G \times M}{U}$ is also closed. To this end we use the monotonicity of function~$\up$ (Item~1. of Lemma~\ref{lem:down}). Let $(v,m_1), (v,m_2)$ be such that $(v,m_2) \in \CPre{0}{G \times M}{U}$ and $(v,m_1) \preceq (v,m_2)$. We have to show that $(v,m_1) \in \CPre{0}{G \times M}{U}$. We first suppose that $v \in V_0$. Hence, by Equation~(\ref{eq:CpreU0}), there exists $((v,m_2),(v',m'_2)) \in \edge{M}$  with $(v',m'_2) \in U$. Recall that $m'_2 = \up(m_2,\p(v))$. Let $m'_1 = \up(m_1,\p(v))$. As $(v,m_1) \preceq (v,m_2)$, by monotonicity of $\up$, it follows that $(v',m'_1) \preceq (v',m'_2)$. Since $U$ is closed, we have $(v',m'_1) \in U$, and thus $(v,m_1) \in \CPre{0}{G \times M}{U}$. The proof is similar if $v \in V_1$.

Let us now prove that $\lceil \CPre{0}{G \times M}{U} \rceil =\lceil C_0 \cup C_1 \rceil = \lceil B_0 \cup B_1 \rceil$ or equivalently that $C_0 = \, \downarrow \! B_0$ and $C_1 = \, \downarrow \! B_1$. We only prove the first equality thanks to Lemma~\ref{lem:down}, the latter one being proved similarly. 

Let $(v,m_1) \in C_0$ with $p = \p(v)$ and let us show that $(v,m_1) \in \, \downarrow \! B_0$. By Equation~(\ref{eq:CpreU0}), there exists $((v,m_1),(v',m'_1)) \in \edge{M}$ with $(v',m'_1) \in U$ and $m'_1 = \up(m_1,p)$. As $A = \lceil U \rceil$, there exists $(v',m'_2) \in A$ such that $(v',m'_1) \preceq (v',m'_2)$. Let $(v,m_2) = (v,\down(m'_2,p)) \in \, \downarrow \! B_0$. By Item~3. of Lemma~\ref{lem:down}, we have $(v,\down(m'_1,p)) \preceq (v,\down(m'_2,p))$. Moreover, as $m'_1 = \up(m_1,p)$, by Item~2. of Lemma~\ref{lem:down}, we have $(v,m_1) \preceq (v,\down(m'_1,p))$. Therefore $(v,m_1) \preceq (v,m_2)$ showing that $(v,m_1) \in \, \downarrow \! B_0$.

Let $(v,m_1) \in \, \downarrow \! B_0$ with $p = \p(v)$ and let us show that $(v,m_1) \in C_0$. As $(v,m_1) \in \, \downarrow \! B_0$, there exists $(v,m_2) \in B_0$ such that $(v,m_1) \preceq (v,m_2)$. As $(v,m_2) \in B_0$, there exists $((v,m_2),(v',m'_2)) \in \edge{M}$ with $(v',m'_2) \in A$ and $m_2 = \down(m'_2,p)$. By Item~1. of Lemma~\ref{lem:down}, we have $(v',\up(m_1,p)) \preceq (v',\up(m_2,p))$. Let $(v',m'_1) = (v',\up(m_1,p))$. By Item~4. of Lemma~\ref{lem:down}, we have $(v',\up(m_2,p)) \preceq (v',m'_2)$. Therefore $(v',m'_1) \preceq (v',m'_2)$. As $(v',m'_2) \in A = \lceil U \rceil$, then $(v',m'_1) \in U$ and thus $(v,m_1) \in C_0$.
\qed\end{proof}

Given a generalized parity game and its extended game $G \times M_1 \times \ldots \times M_k$, we now explain what are the needed adaptations for computing $\JF{0}{G}{\p_1,\ldots,\p_k}$. The approach is similar and works \emph{dimension by dimension} as we did before for parity games. First we define a partial order on $V \times M_1 \times \ldots \times M_k$ such that the partial order of Definition~\ref{def:order} is used on each dimension. More precisely, we define the strict partial order $\prec$ on $V \times M_1 \times \ldots \times M_k$ such that $(v',m'_1,\ldots,m'_k) \prec (v,m_1,\ldots,m_k)$ if and only if $v = v'$ and for all $\ell$:
\begin{enumerate}
\item either $m_\ell, m'_\ell$ are even and $m'_\ell > m_\ell$,
\item or $m_\ell, m'_\ell$ are odd  and $m'_\ell < m_\ell$,
\item or $m_\ell$ is odd and $m'_\ell$ is even.
\end{enumerate}
For instance, if $[d_1] = [d_2] = [3]$, then for a fixed $v$ we have the lower semilattice of Figure~\ref{fig:order}.
\begin{figure}[ht]
	\centering
	\begin{tikzpicture}
		\tikzstyle{nosep}=[inner sep=0pt, outer sep=0pt]
		\matrix (a) [matrix of math nodes, column sep=-1mm, row sep=3mm]{
			& & & (3,3) & & & \\
			& & (3,1) &       & (1,3) & & \\
			& (3,0) & & (1,1) & & (0,3) & \\
			(3,2) & & (1,0) & & (0,1) & &  (2,3)\\
			& (1,2) & & (0, 0) & & (2,1) & \\
			& & (0,2) &       & (2,0) & & \\
			& & & (2,2) & & & \\};
		
		\foreach \i/\j in {1-4/2-3, 1-4/2-5,  2-3/3-2, 2-3/3-4,%
			2-5/3-4, 2-5/3-6, 3-2/4-1, 3-2/4-3, 3-4/4-5, 3-4/4-3,%
			3-6/4-5, 3-6/4-7, 4-1/5-2, 4-3/5-2, 4-3/5-4,4-5/5-4,%
			4-5/5-6, 4-7/5-6, 5-2/6-3, 5-4/6-3, 5-4/6-5,5-6/6-5,6-3/7-4,6-5/7-4}
		\draw (a-\i) -- (a-\j);
	\end{tikzpicture}
	\caption{Partial order $\preceq$}
	\label{fig:order}
\end{figure}
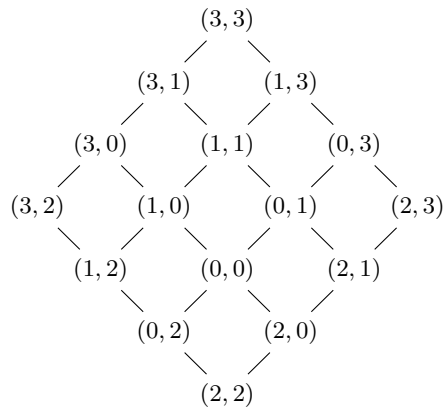

Second, functions $\up_\ell$ and $\down_\ell$, with $\ell \in \{1,\dots, k\}$, are defined exactly as previous functions $\up$ and $\down$ (see Definition~\ref{def:down}), for each dimension $\ell$ and with respect to priority function $\p_\ell$. Third, we adapt (as expected) Proposition~\ref{prop:CPreClosed} for the computation of $\CPre{0}{G \times M_1 \times \ldots \times M_k}{U}$ for a closed set $U \subseteq V \times M_1 \times \ldots \times M_k$. Finally, we obtain an antichain-based algorithm for $\JF{0}{G}{\p_1,\ldots,\p_k}$ as Lemma~\ref{lem:antichain} remains true: each set $T_j$ is a closed set represented by the antichain $\lceil T_j \rceil =  \{(v,0,\ldots,0) \mid v \in F_j\}$.

\section{Algorithms \psolQ\ and \GenpsolQ} \label{sec:psolQ}

In~\cite{HuthKP13}, the authors study another polynomial time partial solver for parity games, called psolQ, that has some similarities with the \psolC\ approach described in Section~\ref{sec:psolC}. It is a more complex partial solver that we present in this section. We then explain how to modify it to get a partial solver for generalized parity games. This partial solver \emph{(i)} works on the initial game structure $G$ and not on the extended game structure $G \times M$, \emph{(ii)} focuses on a \emph{subset $P$ of $i$-priorities} and not on all $i$-priorities, and \emph{(iii)} computes a set similar to $\JF{i}{G}{\p}$ such that the \posAttr\ of Equation~(\ref{eq:mid}) is replaced by a \emph{layered structure of attractors} (one layer per priority $p \in P$).  

\paragraph{\bf Safe attractors.~} 

As we do no longer work with the extended game $G \times M$, we have to adapt the notion \posAttr\ in a way to visit some $i$-priority \emph{while avoiding} visiting any greatest $(1-i)$-priority. Let $G$ be a game structure and $U, U' \subseteq V$ be two subsets of $V$. The \emph{\posSafeAttr} $\PosSafeAttr{i}{G}{U}{U'}$ is the set of vertices from which player $i$ can ensure to visit $U$ in any positive number of steps while visiting no vertex of $U'$. This can be computed with Equations~(\ref{eq:AttrInit}-\ref{eq:Attr}) with the restriction that vertices in $U'$ are not allowed. 

\paragraph{\bf Partial solver for parity games.~}

Let $(G,\Par(\p))$ be a parity game. Let $\pmin$ be some given $i$-priority and $\Pset \subseteq [d]$ be the set composed of all $i$-priorities $p \geq \pmin$. Thus $\Pset  = \{\pmin, \pmin +2, \pmin +4, \ldots, \pmax{i}\}$ such that $\pmax{i}$ is equal to $d$ (resp. $d-1$) if $d$ (resp. $d-1$) is an $i$-priority. Let $U \subseteq \{v \in V \mid \p(v) \in \Pset \}$ be a subset of vertices with an $i$-priority at least $q$. For each $p \in \Pset$, let $U_{p} = U \cap \{v \in V \mid \p(v) \geq p\}$ be the set of vertices of $U$ with priority at least $p$. Hence $U_{\pmax{i}} \subseteq U_{\pmax{i}-2} \subseteq \ldots \subseteq U_{\pmin +2} \subseteq U_{\pmin} = U$. For each $p \in \Pset$, we also consider $U'_{p} = \{v \in V \mid \p(v) \mbox{ is a $(1-i)$-priority and } \p(v) > p\}$, that is, the set of all vertices with $(1-i)$-priority greater than $p$. We compute the following layered structure of \posSafeAttr s:
Initially $B_{\pmax{i} + 2} = \emptyset$ and for all $p \in \Pset$, $B_p$ is computed from $B_{p+2}$ as follows:
\begin{eqnarray}
B_p &=& B_{p+2} \cup \PosSafeAttr{i}{G}{U_{p} \cup B_{p+2}}{U'_{p} \setminus B_{p+2}}.  \label{eq:Bp}
\end{eqnarray}
We call \emph{\layeredAttr} the set $B_{\pmin}$ and we denote it by $\LAttr{i}{G}{\p}{\Pset}{U}$. Notice that $B_{\pmax{i}} \subseteq B_{\pmax{i} - 2} \subseteq \ldots \subseteq U_{\pmin +2} \subseteq B_{\pmin} = \LAttr{i}{G}{\p}{\Pset}{U}$.

For example, consider $d = 9$ and the $0$-priority $q = 4$. Hence $\Pset$ is equal to $\{4,6,\pmax{0} =8\}$. Given some set $U \subseteq \{v \in V \mid \p(v) \in \Pset \}$, $U_8$ (resp. $U_6$, $U_4$) is the set of vertices of $U$ with priority $8$ (resp. priorities in $\{6,8\}$, in $\{4,6,8\}$), and $U'_8$ (resp. $U'_6$, $U'_4$) is the set of vertices with priority $9$ (resp. priorities in $\{7,9\}$, in $\{5,7,9\}$). On Figure~\ref{fig:layer} is depicted the related \layeredAttr\ with three layers.

Let us give some intuition about this layered structure of attractors. From a vertex in $B_{\pmin} \setminus B_{\pmin+2}$ (lowest layer $\pmin$), player~$i$ can ensure to visit $U_{\pmin} \cup B_{\pmin+2}$ without visiting $U'_{\pmin} \setminus B_{\pmin+2}$. In case of a visit to $U_{\pmin}$, this is a good episode for himself (in the sense of Section~\ref{sec:psolC}) since the maximum visited priority along the current history consistent with his strategy is an $i$-priority ($\geq \pmin$). In case of a visit to some $v \in B_{\pmin+2}$, suppose that $v \in B_{p} \setminus B_{p+2}$ (layer $p$ with $p \geq \pmin +2$). Then player~$i$ can now ensure to visit $U_{p} \cup B_{p+2}$ without visiting $U'_{p} \setminus B_{p+2}$. In case of a visit to $U_{p}$, this is again a good episode for player~$i$, otherwise it is a visit to some vertex in a higher layer. If necessary this can be repeated until reaching the highest layer $\pmax{i}$ where player~$i$ can ensure to visit $U_{\pmax{i}}$ without visiting $U'_{\pmax{i}}$ (since $B_{\pmax{i} +2}$ is empty). Thus from all vertices of $B_{\pmin}$, player~$i$ can ensure a good episode for himself. 

\begin{figure}[ht]
	\centering
	\begin{tikzpicture}[scale=0.9]
		\draw (0,7) -- (5,7) -- (5,9.5) -- (0,9.5) -- (0,7);
		\draw (3,7) -- (3,9.5);
		\node[] at (-1,8.25) {Layer 8};
		\node[] at (4,8.25) {$U_8$};
		\node[] at (6,8.25) {$B_8$};
		\draw [->] (1.5,8.25)  -- (3.5,8.25) ;
		
		\draw (0,3.5) -- (5,3.5) -- (5,6) -- (0,6) -- (0,3.5);
		\draw (3,3.5) -- (3,6);
		\draw (3,4.75) -- (5,4.75);
		\node[] at (-1,4.75) {Layer 6};
		\node[] at (4,5.375) {$B_8$};
		\node[] at (4,4.125) {$U_6$};
		\node[] at (6,4.75) {$B_6$};
		\draw [->] (1.5,5.375)  -- (3.5,5.375) ;
		\draw [->] (1.5,4.125)  -- (3.5,4.125) ;
		
		\draw (0,0) -- (5,0) -- (5,2.5) -- (0,2.5) -- (0,0);
		\draw (3,0) -- (3,2.5);
		\draw (3,1.25) -- (5,1.25);
		\node[] at (-1,1.25) {Layer 4};
		\node[] at (4,1.875) {$B_6$};
		\node[] at (4,0.625) {$U_4$};
		\node[] at (6,1.25) {$B_4$};
		\draw [->] (1.5,0.625)  -- (3.5,0.625) ;
		\draw [->] (1.5,1.875)  -- (3.5,1.875) ;
	\end{tikzpicture}	
	\caption{Layered structure of attractors with $U_8 \subseteq U_6 \subseteq U_4 = U$ and $B_8 \subseteq B_6 \subseteq B_4 = \LAttr{0}{G}{\p}{\Pset}{U}$}
	\label{fig:layer}
\end{figure}
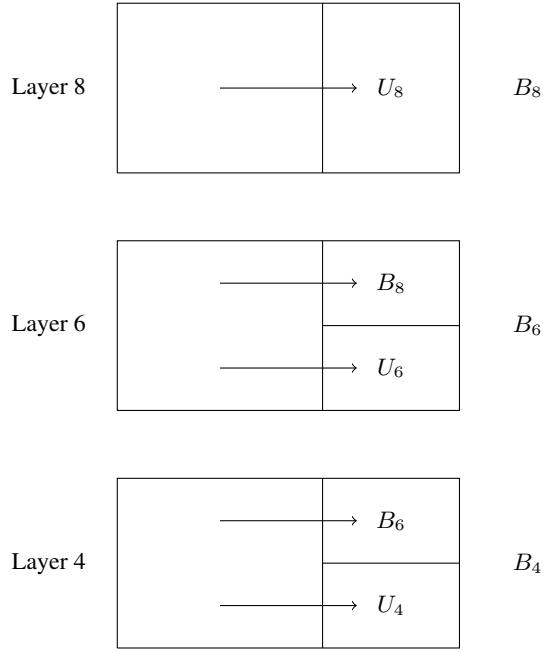

The new partial solver is the same as the \psolC\ solver of Section~\ref{sec:psolC} except that in  Equations~(\ref{eq:begin}-\ref{eq:end}), \emph{(i)} the \layeredAttr\ $\LAttr{i}{G}{\p}{\Pset}{U}$ in the game $G$ replaces the attractor $\PosAttr{i}{G \times M}{T}$ in the extended game $G \times M$ of Equation~(\ref{eq:mid}), and \emph{(ii)} the set $\Pset  = \{\pmin, \pmin +2, \pmin +4, \ldots, \pmax{i}\}$ replaces the set of all $i$-priorities. More precisely, we compute the sequence $(F_j)_{j \geq 0}$ such that $F_0 = \{v \in V \mid \p(v) \in \Pset\}$ and for $j \geq 1$, $F_j$ is computed from $F_{j-1}$ as follows:
\begin{eqnarray}
A_j &=& \LAttr{i}{G}{\p}{\Pset}{F_j} \label{eq:Qmid} \\
F_j &=& A_j \cap F_{j-1}.  \label{eq:Qend}
\end{eqnarray}
We denote by $\QJF{i}{G}{\p}{\Pset}$ the fixpoint $\bigcap_{j\geq 0} F_j$. As in $\JF{i}{G}{\p}$, from any vertex $v_0$ in this fixpoint, player~$i$ can ensure a succession of good episodes implying that he is winning from $v_0$ for $\iPar{i}(\p)$. 

\begin{proposition}[\cite{HuthKP13}] \label{prop:psolQ}
Let $(G,\Par(\p))$ be a parity game and $\Pset \subseteq [d]$ be the set of all $i$-priorities $p \geq q$ for some given $i$-priority $q$. Let $F^{(i)} = \QJF{i}{G}{\p}{\Pset}$. Then for all $v_0 \in V$, if $v_0 \in \Attr{i}{G}{F^{(i)}}$, then $v_0 \in \Win{i}{G}{\iPar{i}(\p)}$. 
\qed\end{proposition}

We derive from this proposition a polynomial time algorithm similar to Algorithm~\ref{algo:psolC} where $\QJF{i}{G}{\p}{\Pset}$ is used instead of $\JF{i}{G}{\p}$ and $\cal P$ is an input list composed of all possible sets $\Pset$, with $q \in [d]$. This algorithm  considers different sets $\Pset$ until finding one composed of $i$-priorities and such that $\QJF{i}{G}{\p}{\Pset}$ is non empty and then makes a recursive call.


\begin{algorithm}
	\caption{\psolQ($G,\p$)}
	\begin{algorithmic}[1]
		\ForEach{$\Pset \in {\cal P}$}
		    \State $i =$ parity of the priorities in $\Pset$
           \State $W = \QJF{i}{G}{\p}{\Pset}$
			\If {$W \neq \emptyset$} 
			    \State $X = \Attr{i}{G}{W}$
			    \State $\{Z_i, Z_{1-i}\} =$ \psolQ($G \setminus X,\p$) 
			    \State return $\{Z_i \cup X, Z_{1-i}\}$
			\EndIf
		\EndFor
		\State return $\{\emptyset,\emptyset\}$ 
	\end{algorithmic}
	\label{algo:psolQ}
\end{algorithm}


\paragraph{\bf Partial solver for generalized parity games.~}

We want to adapt the \psolQ\ approach to generalized parity games. As for the \GenpsolB\ and \GenpsolC\ approaches, we have to treat separately player~$0$ and player~$1$.

For player~$1$, recall that his objective is a disjunction of $\OddPar(\p_\ell)$ over $\ell \in \{1, \ldots, k\}$. Therefore we can apply the \psolQ\ approach to each $\p_\ell$ as explained before for parity games, thus computing $\QJF{1}{G}{\p_\ell}{\Pset}$ for $\Pset$ depending on some $1$-priority $q \in [d_\ell]$. 

For player~$0$, recall that his objective is a conjunction of $\Par(\p_\ell)$ over $\ell \in \{1, \ldots, k\}$. We are going to extend the \psolQ\ approach in this context, for player~$0$. We fix a vector of $0$-priorities $\pmin = (\pmini{1}, \ldots, \pmini{k})$. Let $\pmax{0} = (\pmaxi{1}, \ldots, \pmaxi{k})$ be the vector of maximum $0$-priorities $\pmaxi{\ell} \in [d_\ell]$ for all $\ell$. We denote by $\Pset$ the set $\{\pmax{0}, \pmax{0} - 2, \pmax{0} - 4, \ldots, \pmin\}$ such that if $p = (p_1,\ldots,p_k)$, then $p-2 = (\max\{p_1-2,\pmin_1\},\ldots,\max\{p_k-2,\pmin_k\})$. Let us illustrate with an example. Consider a generalized parity game $d_\ell = 9$ for all $\ell$ and $q =  (4,\ldots,4)$. Then we get $\Pset = \{(8,\ldots,8),(6,\ldots, 6),(4,\ldots,4)\}$. 

Take a vector $p \in \Pset$ and a subset $U \subseteq V$ of vertices.\footnote{In this section, $p,\pmin,\pmax{0}$ are all vectors of priorities (and not a priority).} In a \emph{first step}, let us focus on how player~$0$ can ensure to visit $U$ such that along the history, for all $\ell$, a $0$-priority $\geq p_\ell$ is visited and no $1$-priority $> p_\ell$ is visited (in a way to extend Equation~(\ref{eq:Bp}) temporarily without set $B_{p+2}$). Such a generalized reachability can be reduced to reachability by working with the game $G$ extended with memory ${\cal N}_p$ such that $N \in {\cal N}_p$ records the dimensions $\ell$ for which a vertex with $0$-priority $\geq p_\ell$ is already visited. More precisely, we consider the extended game structure $G_p = G \times {\cal N}_p$ with ${\cal N}_p = \{N \mid N \subseteq \{1,\ldots, k\}\}$. For each $v \in V$, we denote by $N_p(v) = \{\ell \mid \p_\ell(v) \mbox{ is a $0$-priority such that } \p_\ell(v) \geq p_\ell \}$. The game structure $G_p$ has the set $V \times {\cal N}_p$ of vertices and the set $E_p$ of edges $((v,N),(v',N'))$ such that $(v,v') \in E$ and $N' = N \cup N_p(v')$.\footnote{Notice that even if ${\cal N}_p$ does not depend on $p$, the edges of $G_p$ depend on $p$.} Moreover, $V_i \times {\cal N}_p$ is the set of vertices controlled by player~$i$, for $i \in \{0,1\}$. Finally to any initial vertex $v_0$ in $G$ corresponds the initial vertex $(v_0,N_p(v_0))$ in the extended game $G_p$.

Define $T_p = \{ (v,N) \mid v \in U, N = \{1,\ldots,k\} \}$ and $T'_p = \{ (v,N) \in  V \times {\cal N}_p \mid \exists \ell, \, \p_\ell(v)$ is a $1$-priority $> p_\ell \}$\footnote{Notice that $T_p$ depends on $U$ while $T_p$ does not.}. Clearly, in game $G$, player~$0$ can ensure to visit $U$ such that for all $\ell$, a $0$-priority $\geq p_\ell$ is visited and no $1$-priority $> p_\ell$ is visited from $v_0$ if and only if in the extended game $G_p$, he can ensure to visit $T_p$ while not visiting $T'_p$ from $(v_0,N_p(v_0))$. The latter condition is equivalent to  $(v_0,N_p(v_0) \in \PosSafeAttr{0}{G_p}{T_p}{T'_p}$. In this case player~$0$ can derive from $G_p$ to $G$, a winning strategy from $v_0$ (that has a finite memory equal to ${\cal N}_p$).

In a \emph{second step}, let us explain how to generalize the concept of \layeredAttr, and in particular how to manage the set $B_{p+2}$ in Equation~(\ref{eq:Bp}). For parity games we explained how player~$0$ has to adapt his attractor strategy when he shifts from layer~$p$ to some higher layer $p' > p$ due to the visit to some $v \in B_{p+2}$. Here when player~$0$ visits some vertex $(v,N)$ in layer $p$ for which he has to shift to layer $p'$, he stops applying his winning strategy for layer~$p$, and begins applying his winning strategy for layer~$p'$ from the initial vertex $(v,N_{p'}(v))$ belonging to layer~$p'$. Hence Player~$0$ thus \emph{resets} his memory from $(v,N)$ to $(v,N_{p'}(v))$.

We are now ready to adapt Equation~(\ref{eq:Bp}) to compute a \layeredAttr\ in the context of generalized parity games. Recall that $\Pset = \{\pmax{0}, \pmax{0} - 2, \ldots, \pmin\}$ for some given vector $\pmin = (\pmini{1}, \ldots, \pmini{k})$ of $0$-priorities and that $U$ is a subset of $V$. We compute the following layered structure of \posSafeAttr s:
initially $C_{\pmax{0} + 2} = \emptyset$ and for all $p \in \Pset$, $C_p$ is computed from $C_{p+2}$ as follows:
\begin{eqnarray*}
B_p &=& \PosSafeAttr{0}{G_p}{T_{p} \cup C_{p+2}}{T'_{p} \setminus C_{p+2}}, \\
C_p &=& C_{p+2} \cup \{ (v,N) \mid N \subseteq \{1,\ldots,k\}, (v,N_p(v)) \in B_p\} 
\end{eqnarray*}
The \layeredAttr\ is here the set $\{v \in V \mid (v,N_p(v)) \in B_{\pmin}\}$; we denote it by $\LAttr{0}{G}{\p_1,\ldots,\p_k}{\Pset}{U}$.

In a \emph{last step}, it remains to explain how to use this adapted notion of \layeredAttr\ to obtain a partial solver for generalized parity games. We simply take the same equations (\ref{eq:Qmid}-\ref{eq:Qend}) except that $\LAttr{0}{G}{\p}{\Pset}{F_j}$ must be replaced by $\LAttr{0}{G}{\p_1,\ldots,\p_k}{\Pset}{F_j}$. We denote by $\QJF{0}{G}{\p_1,\ldots,\p_k}{\Pset}$ the fixpoint $\bigcap_{j\geq 0} F_j$. From all the given arguments in the previous three steps follows the next proposition.

\begin{proposition}[\cite{HuthKP13}] 
Let $(G,\ConjPar(\p_1,\ldots,\p_k))$ be a generalized parity game. For some given vector $\pmin = (\pmini{1}, \ldots, \pmini{k})$ of $0$-priorities, let $\Pset = \{\pmax{0}, \pmax{0} - 2, \ldots, \pmin\}$ and let $F^{(0)} = \QJF{0}{G}{\p_1,\ldots,\p_k}{\Pset}$. Then for all $v_0 \in V$, if $v_0 \in \Attr{0}{G}{F^{(0)}}$, then $v_0 \in \Win{0}{G}{\ConjPar(\p_1,\ldots,\p_k)}$. 
\qed\end{proposition}

We derive from this proposition an algorithm (see Algorithm~\ref{algo:GenpsolQ}) similar to Algorithm~\ref{algo:GenpsolC} such that for each dimension $\ell$, $\JF{1}{G}{\p_\ell}$ is replaced by $\QJF{1}{G}{\p_\ell}{\Pset}$ with $\Pset$ depending on some $1$-priority $q \in [d_\ell]$ (player~$1$) and $\JF{0}{G}{\p_1,\ldots,\p_k}$ is replaced by $\QJF{0}{G}{\p_1,\ldots,\p_k}{\Pset}$ with $\Pset$ depending on some vector $q$ of $0$-priorities (player~$0$). However, as for Algorithm~\ref{algo:psolQ}, this algorithm needs an input list $\cal P$ composed of: 
\begin{itemize}
    \item elements $(\Pset,\ell)$ with $\Pset \subseteq [d_\ell]$ such that $\pmin$ is a $1$-priority in $\{\pmax{1}, \pmax{1} - 2, \ldots, 1\}$,
    \item elements $\Pset \subseteq [d_1] \times \ldots \times [d_k]$ such that $\pmin$ is a vector of $0$-priorities in $\{\pmax{0}, \pmax{0} - 2, \ldots, 0\}$. 
\end{itemize}
By Theorem~\ref{thm:complexity}, this algorithm is in $O((\max_{\ell=1}^{k}\frac{d_\ell}{2})^2 \cdot |V|^2 \cdot |E| \cdot 2^k)$ time complexity.

\begin{algorithm}
	\caption{\GenpsolQ($G,\p_1,\ldots,\p_k, \GenList$)}
	\begin{algorithmic}[1]
		\ForEach{$element \in \GenList$}
		\If {$element = (\Pset,\ell)$}
			\State $W = \QJF{1}{G}{\p_\ell}{\Pset}$  
			\If {$W \neq \emptyset$} 
			    \State $X = \Attr{1}{G}{W}$
			    \State $\{Z_0,Z_1\} =$ \GenpsolQ($G \setminus X, \p_1,\ldots,\p_k, \GenList$)
		        \State return $\{Z_0,Z_1 \cup X\}$	
		    \EndIf
		\Else {\quad \quad We know that $element = \Pset$ with $\Pset \subseteq [d_1] \times \ldots \times [d_k]$}
		   \State $W = \QJF{0}{G}{\p_1,\ldots,\p_k}{\Pset}$ 
		    \If {$W \neq \emptyset$} 
		        \State $X = \Attr{0}{G}{W}$
		        \State $\{Z_0,Z_1\} =$ \GenpsolQ($G \setminus X,\p_1,\ldots,\p_k, \GenList$)
		        \State return $\{Z_0 \cup X,Z_1\}$
		    \EndIf
		\EndIf
		\EndFor
		\State return $\{\emptyset,\emptyset\}$ 
	\end{algorithmic}
	\label{algo:GenpsolQ}
\end{algorithm}





\section{Empirical Evaluation}
\label{sec:experiments}

For parity games, the polynomial time partial solvers \psolB\footnote{The variant with safety objectives.}, \psolC, and \psolQ\ are theoretically compared in~\cite{HuthKP13,HuthKP16}. It is proved that the partial solutions computed by Algorithm \psolB\ are included in the partial solutions computed by Algorithm \psolQ\ themselves included in the partial solutions computed by Algorithm \psolC. Examples of parity games are also given that distinguish the three partial solvers (strict inclusion of partial solutions), as well as an example that is not completely solved by the most powerful partial solver \psolC. This behavior also holds for the three partial solvers proposed here for generalized parity games. Furthermore, their time complexity is exponential in the number $k$ of priority functions whereas the classical algorithm for generalized parity games is exponential in both $k$ and all $d_\ell$~\cite{ChatterjeeHP07}.  

For both parity games and generalized parity games, we implemented the three partial solvers (with the antichain approach for Algorithm~\GenpsolC), Algorithm \Ziel\ (resp. \GenZiel) and its combination \PartialZiel\ (resp. \GenPartialZiel) with each partial solver, and we executed all these algorithms on a large set of benchmarks. 

\paragraph{\bf Setting.~}
Our benchmarks were generated from TLSF specifications used for the Reactive Synthesis Competition (SYNTCOMP~\cite{syntcomp18}) using a compositional translation as explained in the introduction.\footnote{The tool we implemented to realize this translation can be fetched from \url{https://github.com/gaperez64/tlsf2gpg}} The source
code for our prototype tool along with all the information about our benchmarks
is publicly available at \url{https://github.com/Skar0/generalizedparity}. Our experiments have been carried out on a server with Mac OS X 10.13.4 (build 17E199). As hardware, the server had as CPU one 6-Core Intel Xeon; as processor speed, 3.33 GHz; as L2 Cache (per Core), 256 KB; as L3 Cache, 12 MB; as memory, 32 GB; and as processor interconnect speed, 6.4 GT/s. We implemented our algorithms in Python 2.7.

\paragraph{\bf Experiments on parity games.~}

We considered 240 benchmarks for parity games. Those games have a mean size $|V|$ around 46K with a maximal size of 3157K, and a mean number $d$ of priorities of $4.1$ with a maximal number $d= 15$. The statistics about the behaviors of the different algorithms are summarized in Table~\ref{tab:onedim1} (for all the 240 benchmarks) and Table~\ref{tab:onedim2} (for the 20 most difficult benchmarks for Algorithm~\Ziel). Column~$1$ indicates the name of the solver, Column~$2$ counts the number of benchmarks completely solved (for the partial solvers, the second number is the number of incomplete solutions), Column~$3$ counts the number of timeouts (fixed at 60000 ms), and Column~$4$ counts how many times the solver was the fastest (excluding examples with timeout). In Table~\ref{tab:onedim1} (resp.~\ref{tab:onedim2}), for the 233 (resp. 13) benchmarks without timeout for all the complete Algorithms \Ziel\ and \PartialZiel, Column~$5$ indicates the mean execution time in milliseconds.

\begin{table}
\begin{center}
\begin{tabular}{|c| c| c| c|c|}
    \hline
     Solver & Solved & T.O. & Fastest & Mean time (233)   \\
     \hline
     \Ziel & 240 (100\%) & 0 & 150 (63 \%) & 272 ms         \\  
     \hline
     Ziel\&\psolB & 240 (100\%) & 0 & 89 (37 \%) & 480 ms   \\
     \hline
     Ziel\&\psolC & 233 (97\%) & 7 & 0 (0\%) & 1272 ms      \\
     \hline
     Ziel\&\psolQ & 238 (99\%) & 2 & 1 (0\%) & 587 ms     \\
     \hline
     \psolB & 203 (84\%) - 37 & 0  & - & -                        \\
          \hline
      \psolC & 233 (97\%) - 0 & 7 & - & -                        \\
     \hline
      \psolQ & 232 (97\%) - 6 & 2 & - & -                         \\
     \hline
\end{tabular}%
\end{center}
\caption{\label{tab:onedim1} Statistics on the one dimensional benchmarks.}
\end{table}

\begin{table}
\begin{center}
\begin{tabular}{|c| c| c| c|c|}
    \hline
     Solver & Solved & T.O. & Fastest & Mean time (13) \\
     \hline
     \Ziel & 20 (100\%) & 0 & 11 (55\%) & 451 ms  \\  
     \hline
     Ziel\&\psolB & 20 (100\%) & 0 & 8 (40\%) & 7746 ms  \\
     \hline
     Ziel\&\psolC & 13 (65\%) & 7 & 0 (0\%) & 20025 ms   \\
     \hline
     Ziel\&\psolQ & 18 (99\%)  & 2 & 1 (5\%) & 9079 ms  \\
     \hline
     \psolB & 15 (75\%) - 5 & 0  & - & -\\
     \hline
      \psolC & 13 (65\%) - 0 & 7 & - & -\\
     \hline
    \psolQ & 18 (90\%) - 0 & 2 & - & -\\
      \hline
\end{tabular}
\end{center}
\caption{\label{tab:onedim2} Statistics on the one dimensional benchmarks for the 20 most difficult instances for Algorithm \Ziel}
\end{table}

\paragraph{\bf Experiments for generalized parity games.~}

We considered 152 benchmarks for generalized parity games. Those games have a mean size $|V|$ around 207K with a maximal size of 7009K. The mean number of priority functions is equal to $4.53$ with a maximum number of $17$. The statistics about the behaviors of the different algorithms are summarized in Table~\ref{tab:severaldim}. The columns have the same meaning as before and the last column concerns the 87 benchmarks without timeout for all Algorithms \GenZiel\ and \GenPartialZiel.

\begin{table}[]
\begin{center}
\begin{tabular}{|c|c|c|c|c|}
\hline
Solver & Solved & T.O.  &     Fastest               &       Mean-Time  (87)          \\ \hline
 \GenZiel          & 128 (84\%) & 24 &     33 (25\%)              &         66 ms          \\ \hline
 GenZiel\&\GenpsolB & 130 (86\%) & 22 &     72  (55\%)             &         56 ms         \\ \hline
 GenZiel\&\GenpsolC & 112 (74\%) & 40 &     24   (18\%)            &         644 ms             \\ \hline
  GenZiel\&\GenpsolQ & 110 (72\%) & 42 &     3    (2\%)           &          1133 ms         \\ \hline
 \GenpsolB &         110 (72\%) - 20 & 22 &         -             &       -         \\ \hline
 \GenpsolC &         112 (74\%) - 0 & 40 &        -               &      -            \\ \hline
  \GenpsolQ &         104 (68\%) - 6 & 42 &         -              &       -         \\ \hline
\end{tabular}
\end{center}

\caption{Statistics on the multi-dimensional benchmarks.}
\label{tab:severaldim}
\end{table}

\paragraph{\bf Observations.~}
Our experiments show that for parity games, Algorithm~\Ziel\ is faster than partial solvers on average which was not observed on random graphs in~\cite{HuthKP16}. For generalized parity games, they show that $4$ benchmarks that cannot be solved by Algorithm \GenZiel\ or by a partial solver alone, can be solved by the combination of \GenZiel\ with a partial solver. Our experiments also show that their combination with a partial solver improves the performances of Algorithms \Ziel\ and \GenZiel\ on a large portion of the benchmarks: 90 cases over 240 (38\%) for parity games and 99 cases over 132 (75\%) for generalized parity cases. They suggest that it is interesting to launch in parallel Algorithms GenZiel\&\GenpsolB, GenZiel\&\GenpsolC, and GenZiel\&\GenpsolQ, as none appears to dominate the other ones. Nevertheless, the combination of \GenZiel\ with \GenpsolB\ is a good compromise.

\section{Conclusion}

In this paper, we have shown how to extend the three partial solvers for parity games proposed in~\cite{HuthKP13,HuthKP16} to the case of generalized parity games (conjunction of parity conditions).
For one of those partial solver, we also have provided an antichain-based variant to retain efficiency. In addition, we have shown how to combine those partial solvers with the classical recursive algorithms due to Zielonka~\cite{zielonka98} for parity games and its extension~\cite{ChatterjeeHP07} for generalized parity games. For both parity games and generalized parity games, we have implemented the classical recursive algorithm, the three partial solvers and their combinations with the classical algorithm. All these algorithms have been tested on a large set of instances of parity games and generalized parity games, that were generated from meaningful LTL specifications proposed in the last LTL reactive synthesis competition~\cite{syntcomp18}. For parity games, our  experiments show that Zielonka’s algorithm is faster than partial solvers on average. This was not observed on random graphs in~\cite{HuthKP16} where Algorithm \psolB\ was faster. Equally interestingly, it appears that for generalized parity games, the combination of Algorithm \GenZiel\ with a partial solver allows to solve instances that cannot be solved by the classical recursive algorithm or a partial solver alone.  Finally for both the parity games and the generalized parity games, our experiments indicate that the performances of the classical recursive algorithms are often improved when combined with partial solvers.

\bibliographystyle{abbrv}
\bibliography{biblio}

\end{document}